\documentclass[american,french,english]{paper}
\usepackage[T1]{fontenc}
\usepackage[latin9]{inputenc}
\usepackage{refstyle}
\usepackage{float}
\usepackage{textcomp}
\usepackage{amsthm}
\usepackage{amsmath}
\usepackage{amssymb}
\usepackage{fixltx2e}
\usepackage{graphicx}
\usepackage{esint}

\makeatletter

\AtBeginDocument{\providecommand\lemref[1]{\ref{lem:#1}}}
\AtBeginDocument{\providecommand\propref[1]{\ref{prop:#1}}}
\providecommand{\tabularnewline}{\\}
\RS@ifundefined{subref}
  {\def\RSsubtxt{section~}\newref{sub}{name = \RSsubtxt}}
  {}
\RS@ifundefined{thmref}
  {\def\RSthmtxt{theorem~}\newref{thm}{name = \RSthmtxt}}
  {}
\RS@ifundefined{lemref}
  {\def\RSlemtxt{lemma~}\newref{lem}{name = \RSlemtxt}}
  {}

 \newcommand\thmsname{\protect\theoremname}
 \newcommand\nm@thmtype{theorem}
 \theoremstyle{plain}
 
 \newenvironment{namedthm}[1][Undefined Theorem Name]{
   \ifx{#1}{Undefined Theorem Name}\renewcommand\nm@thmtype{theorem*}
   \else\renewcommand\thmsname{#1}\renewcommand\nm@thmtype{namedtheorem}
   \fi
   \begin{\nm@thmtype}}
   {\end{\nm@thmtype}}
\theoremstyle{plain}
\newtheorem{thm}{\protect\theoremname}[section]
  \theoremstyle{remark}
  \newtheorem{rem}[thm]{\protect\remarkname}
  \theoremstyle{plain}
  \newtheorem{prop}[thm]{\protect\propositionname}
\newenvironment{lyxlist}[1]
{\begin{list}{}
{\settowidth{\labelwidth}{#1}
 \setlength{\leftmargin}{\labelwidth}
 \addtolength{\leftmargin}{\labelsep}
 }}
{\end{list}}
  \theoremstyle{plain}
  \newtheorem{lem}[thm]{\protect\lemmaname}

\date{07/02/2014}

\makeatother

\usepackage{babel}
\makeatletter
\addto\extrasfrench{%
   \providecommand{\fg}{\ifdim\lastskip>\z@\unskip\fi~\frqq}%
}

\makeatother
  \addto\captionsamerican{\renewcommand{\lemmaname}{Lemma}}
  \addto\captionsamerican{\renewcommand{\propositionname}{Proposition}}
  \addto\captionsamerican{\renewcommand{\remarkname}{Remark}}
  \addto\captionsamerican{\renewcommand{\theoremname}{Theorem}}
  \addto\captionsenglish{\renewcommand{\lemmaname}{Lemma}}
  \addto\captionsenglish{\renewcommand{\propositionname}{Proposition}}
  \addto\captionsenglish{\renewcommand{\remarkname}{Remark}}
  \addto\captionsenglish{\renewcommand{\theoremname}{Theorem}}
  \addto\captionsfrench{\renewcommand{\lemmaname}{Lemme}}
  \addto\captionsfrench{\renewcommand{\propositionname}{Proposition}}
  \addto\captionsfrench{\renewcommand{\remarkname}{Remarque}}
  \addto\captionsfrench{\renewcommand{\theoremname}{Th�or�me}}
  \providecommand{\lemmaname}{Lemma}
  \providecommand{\propositionname}{Proposition}
  \providecommand{\remarkname}{Remark}
  \providecommand{\theoremname}{Theorem}
\providecommand{\theoremname}{Theorem}

\begin{document}
\title{A Tracking Approach to Parameter Estimation in Linear Ordinary Differential
Equations}

 \author{{\large Nicolas J-B. Brunel, Quentin Clairon}}
\institution{ENSIIE \& Laboratoire de Math\'ematiques et Mod\'elisation d'Evry,\\  UMR CNRS 8071, Universit\'e
d'Evry, France}

\maketitle
\begin{abstract}
Ordinary Differential Equations are widespread tools to model chemical,
physical, biological process but they usually rely on parameters which
are of critical importance in terms of dynamic and need to be estimated
directly from the data. Classical statistical approaches (nonlinear
least squares, maximum likelihood estimator) can give unsatisfactory
results because of computational difficulties and ill-posedness of the 
 statistical problem. New estimation methods that use some nonparametric devices have been proposed
to circumvent these issues. We present a new estimator that shares properties with 
 Two-Step estimator and Generalized Smoothing (introduced by Ramsay et al. \cite{Ramsay2007}). We introduce a perturbed model and we use optimal control theory for constructing a criterion that aims at minimizing the discrepancy with
data and the model. Here, we focus on the case of linear Ordinary Differential Equations as our criterion has a closed-form expression that permits a detailed analysis.  
 Our approach avoids the use of a nonparametric estimator
of the derivative, which is one of the main cause of inaccuracy in Two-Step estimators. Moreover, we take into account model discrepancy and our estimator is more robust to model misspecification than classical methods. The discrepancy with the parametric ODE model correspond to the minimum perturbation (or control) to apply to the initial
model. Its qualitative analysis can be informative for misspecification
diagnosis. In the case of well-specified model, we show the consistency
of our estimator and that we reach the parametric $\sqrt{n}-$ 
rate when regression splines are used in the first step.
\end{abstract}

%Inference; Asymptotics; Ordinary Differential Equations; Optimal Control; Generalized Smoothing

\section{Introduction}

We consider a dynamical process defined by an Ordinary Differential
Equation (ODE) with a known and fixed initial value

\begin{equation}
\left\{ \begin{array}{ll}
\dot{x} & =f(t,x,\theta)\\
x(0) & =x_{0}
\end{array}\right.\label{eq:ode1}
\end{equation}
 Such a model is called an Initial Value Problem (IVP). The state
$x$ is in $\mathbb{R}^{d}$ and $\theta$ is an unknown parameter,
that belongs to a subset $\Theta$ of $\mathbb{R}^{p}$. $f$ is a
time-dependent vector field from $\left[0,\, T\right]\times\mathbb{R}^{d}\times\Theta$
to $\mathbb{R}^{d}$. This class of dynamical models are commonly
used in physics, engineering, ecology,\dots \cite{Fuguitt1947,Mirsky2009,CellComputl2002,albert1997biochemical}.
Let $t\mapsto X_{\theta^{*}}(t)=X^{*}(t)$ be the solution to the
IVP (\ref{eq:ode1}) on $\left[0,T\right]$, for the true parameter
set $\theta^{*}$. 

We want to estimate $\theta^{*}$ from noisy observations $Y_{i},\: i=1,\dots,n$
of the trajectory $X^{*}$, made at time $t_{i}$. Estimation can
be done by classical estimators such as Nonlinear Least Squares (NLS),
Maximum Likelihood Estimator (MLE) \cite{LiOsborne2005} or Bayesian
approaches (\cite{Huang2006},\cite{gelman1996},\cite{calderhead2009}
and \cite{ghasemi2011} for example). Nevertheless, the statistical
estimation of an ODE model by NLS leads to a difficult nonlinear estimation
problem. These difficulties were pointed out by Ramsay et al. \cite{Ramsay2007}:
computational complexity comes from repeated ODE integrations required
by the optimization algorithm; moreover, the usual criterion exhibits
multiple local minima . Even though meta-heuristic methods can be
proposed to circumvent the last issue, parameter estimation can be
considered as an ill-posedness inverse problem \cite{Engl2009}, that
needs alternatives to the classic statistical approaches.

Alternative statistical estimators have been developped or adapted
to this particular framework, such as hierarchical Bayesian approaches
\cite{Huang2006,Raftery2010} and MCMC (\cite{calderhead2011}), Generalized
Smoothing \cite{Ramsay2007,QiZhao2010,Campbell2011,Campbell2013}
or Two-Step estimators \cite{Brunel2008,Brunel2013,LiangWujasa2008,GugushviliKlaassen2010}.
Recently variational approaches have also been developed \cite{Timmer2012}. 

As other two-step estimators, our method produces a minimum distance
estimator \cite{Koul1992} but it shares strong links with Generalized
Smoothing approaches. Two-Step estimators were initiated by \cite{Varah1982}
and aims at minimizing a discrepancy measure between a nonparametric
estimator $\hat{X}$ and quantities characterizing the differential
models. Usually, a Two-Step method is defined by the following procedure:
\begin{enumerate}
\item Construct a nonparametric curve estimator $\widehat{X}$ from the
data $(t_{i},Y_{i})_{1\leq i\leq n}$, 
\item Compute a model discrepancy measure $R(\hat{X},\theta)$, such as
the weighted $L^{2}$ distance 
\begin{equation}
R(\hat{X},\theta)=\int_{0}^{T}\left\Vert \dot{\widehat{X}}(t)-f(t,\widehat{X}(t),\theta)\right\Vert ^{2}w(t)dt\label{eq:classicalTS}
\end{equation}

\item Define the parameter estimator 
\begin{equation}
\widehat{\theta}^{TS}=\arg\min_{\theta\in\Theta}R(\hat{X},\theta)\label{eq:TS-Estimator}
\end{equation}

\end{enumerate}
These estimators have a good computational efficiency vs NLS and they
avoid repeated ODE integration. In practice, the used criteria are
also smoother and easier to optimize than the NLS criterion. Two-step
estimators are consistent in general, but there is a trade-off with
the statistical precision, and some care in the use of nonparametric
estimate $\dot{\widehat{X}}$ has to be taken in order to keep the parametric
rate \cite{Brunel2008,GugushviliKlaassen2010}. The variance of TS
estimator can be higher, but the use of other criterion, for instance
based on a weak formulation of the differential equation can give
competitive alternatives, in particular in high dimensional parameter
space or small sample size, see \cite{Brunel2013}.

 In the case of Generalized Smoothing \cite{Ramsay2007}, the solution
$X^{*}$ is approximated by a basis expansion that solves approximately
the ODE model; hence, the parameter inference is performed by dealing
with an imperfect model, as the collocation approximation of the ODE
solution can be seen as a relaxation of the ODE model constraint, needed
for taking into account some uncertainty about the model. Based on
the Generalized Profiling approach, Hooker proposed a criteria that
estimates the lack-of-fit through the estimation of a ``forcing function''
$t\mapsto u(t)$ in the ODE $\dot{x}-f(t,x,\tilde{\theta})=u(t)$,
where $\tilde{\theta}$ is a previous estimate obtained by Generalized
Profiling \cite{Brunel2013}.  The objective of this paper is to
provide a parameter estimate and an approximate solution to the ODE that 
\begin{itemize}
\item avoids the use a nonparametric estimate of the derivative $\dot{X}$ as in two step estimators, 
\item incorporates robustness in model specification and controls the quality of approximation, 
\item introduces the use of infinite dimensional optimization tools, exploiting the differential structure of the model.
\end{itemize}
One interest of the latter point is to avoid the use of series expansion
for function estimation, and avoid some arbitrary practical choices about the
basis. Moreover, infinite dimensional optimization tools
give powerful characterization of the solutions that gives an additional
insight in Generalized Smoothing.

Our method provides a consistent parametric estimator when the model
is correct. We show that it is root-$n$ consistent and asymptotically
normal. At the same time, we get a discrepancy measure between the
model and the data under the form of an optimal control $u$ analogous
to the forcing function in \cite{Hooker2009}. 

In the next section, we introduce the notations and we motivate our
approach by discussing the Generalized Smoothing approach, and the
link with Optimal Control Theory. In section \ref{sec:Consistency-of-the},
we show that the estimator is consistent under some regularity assumption
about the model. Then in section \ref{sec:Asymptotics-of}, we show
that we reach the root$-n$ rate using regression splines for $\widehat{X}$.
Finally, we show the interest of our method on a toy model (and we
make comparison with Nonlinear Least Squares and Generalized Smoothing) and we consider also a real data case, where a linear ODE is used for describing the isomerization reaction of $\alpha$-Pinene.

\section{Model and methodology\label{sec:Model-and-methodology}}

We introduce our statistical framework, and we recall the mechanics
of the Generalized Smoothing estimator in the particular context of
a linear ODE.

\subsection{The statistical model and a Generalized Smoothing wrap-up }

We observe a ``true'' trajectory $X^{*}$ at $n$ random times $0=t_{1}<t_{2}\dots<t_{n}=T$,
such that we have $n$ observations $(Y_{1},\dots,Y_{n})$ defined
as 
\begin{equation}
Y_{i}=X^{*}(t_{i})+\epsilon_{i}\label{eq:StatModel}
\end{equation}
where $\epsilon_{i}$ is the (random) observation error. We assume
that there is a true parameter $\theta^{*}$ belonging to a subset
$\Theta$ of $\mathbb{R}^{p}$, such that $X^{*}$ is the unique solution
of the linear ODE 
\begin{equation}
\dot{x}(t)=A_{\theta}(t)x(t)+r_{\theta}(t)\label{eq:LinearODEmodel}
\end{equation}
with initial condition $X^{*}(0)=X_{0}^{*}$; where $t\mapsto A_{\theta}(t)\in\mathbb{R}^{d\times d}$
and $t\mapsto r_{\theta}(t)\in\mathbb{R}^{d}$. More generally, we
denote $X_{\theta}$ the solution of (\ref{eq:LinearODEmodel}) for
a given $\theta$, and initial condition $X_{0}^{*}$. We assume that
the initial condition $X_{0}^{*}$ is exactly known, and we want 
to infer $\theta^{*}$ from $(Y_{1},\dots,Y_{n})$. In the linear
case, Duhamel's formula gives a closed form expression for $X_{\theta}$
for $t\in\left[0,T\right]$ 
\[
X_{\theta}(t)=\Phi_{\theta}(t,0)X_{0}^{*}+\int_{0}^{T}\Phi_{\theta}(t,s)r_{\theta}(s)ds
\]
where the matrix-valued function $\Phi_{\theta}\,:\,\left(t,s\right)\mapsto\Phi_{\theta}\left(t,s\right)$
is the so-called resolvant of the ODE. By definition, the resolvant
is the solution of the homogenous ODE
\[
\begin{array}{l}
\dot{\Phi_{\theta}}(t,s)=A_{\theta}(t)\Phi_{\theta}(t,s)\\
\Phi_{\theta}(s,s)=I_{d}
\end{array}
\]
The estimation of $\theta^{*}$ can be done straightforwardly  with
the Nonlinear Least Squares (NLS) estimator that minimizes 
\[
\sum_{i=1}^{n}\left\Vert Y_{i}-X_{\theta}(t_{i})\right\Vert _{2}^{2}.
\]
 In Generalized Smoothing (GS), parameter estimation is regularized
by using an approximate solutions of the ODE (\ref{eq:LinearODEmodel}),
as GS takes advantage of the double interpretation of splines for
smoothing data, and for numerical solving of ODE by collocation.

A basis expansion $\widehat{X}(t,\theta)=\widehat{\beta}(\theta)^{T}p(t)$
is computed for each $\theta$, where $\widehat{\beta}(\theta)$ is
obtained by minimizing in $\beta$ the criterion 

\begin{equation}
J_{n}(\beta\vert\theta,\lambda)=\sum_{i=1}^{n}\left\Vert Y_{i}-\widehat{\beta}^{T}p(t_{i})\right\Vert _{2}^{2}+\lambda\int_{0}^{T}\left\Vert \widehat{\beta}^{T}\dot{p}(t)-\left(A_{\theta}(t)\widehat{\beta}^{T}p(t)+r_{\theta}(t)\right)\right\Vert _{2}^{2}dt\label{eq:Cost_GS}
\end{equation}
This first step is considered as profiling along the nuisance parameter
$\beta$, whereas the estimation of the parameter of interest is obtained
by minimizing the sum of squared errors of the proxy $\hat{X}(t,\theta)$
by 
\begin{equation}
\hat{\theta}^{GS}=\arg\min_{\theta}\sum_{i=1}^{n}\left\Vert Y_{i}-\hat{X}(t_{i},\theta)\right\Vert ^{2}\label{eq:Estimator_GS}
\end{equation}
Obviously, the estimator depends on the hyperparameter $\lambda$,
that needs to be selected from the data in practice (some adaptive
procedures have been proposed, see \cite{chervoneva2014}). The essential
difference with NLS is to replace the exact solution $X_{\theta}(\cdot)$
by an approximation $\hat{X}(\cdot,\theta)$ (that depends also on
the data). This means that GS deals with 2 sources of errors: in addition to the classical statistical error (variance due to noisy data), there is an approximation error as $\hat{X}(\cdot,\theta)$ is a spline that does not solve exactly the ODE model (\ref{eq:LinearODEmodel}).
Indeed, collocation algorithms compute the coefficients of a B-spline
expansion based on the relationships between $\hat{X}$ and its derivative
$\dot{\hat{X}}$ evaluated on an appropriate grid of time points $0=s_{1}<s_{2}<\dots<s_{p}$,
\cite{boor2001}. This gives a nonlinear system that is usually
solved with a Newton algorithm, whose roots are the unknown coefficients
of the basis expansion. The collocation schemes are essentially useful
for solving Boundary Value Problems (instead of the classical Initial
Value Problem). 

For parameter estimation, the basis expansion is defined in a somehow
arbitrary manner and the ODE constraint is not used as an equality constraint as it should be the
case in a ``normal'' collocation scheme. Instead, the ODE equation
is transformed into an inequality constraint defined on the interval
$[0,T]$ and the model constraint is never set to 0 because of the
trade-off with the data-fitting term $\sum_{i=1}^{n}\left\Vert Y_{i}-\widehat{\beta}^{T}p(t_{i})\right\Vert _{2}^{2}$.
For this reason, the ODE model (\ref{eq:LinearODEmodel}) is not solved
and it is useful to introduce the discrepancy term $\hat{u}_{\theta}(t)=\widehat{\beta}^{T}\dot{p}(t)-\left(A_{\theta}(t)\widehat{\beta}^{T}p(t)+r_{\theta}(t)\right)$
that corresponds to a model error. In fact, the proxy $\hat{X}(\cdot,\theta)$
satisfies a pertubed ODE $\dot{x}=A_{\theta}x+r_{\theta}+\hat{u}_{\theta}$.
This forcing function $\hat{u}_{\theta}$ is an outcome of the optimization
process and can be relatively hard to analyze or understand, as it
depends on the basis expansion used and it depends also on the data
via the minimization of $J_{n}(\beta\vert\theta,\lambda)$. Nevertheless,
Hooker et al. have proposed goodness-of-fit tests based on this so-called
``empirical forcing function'' $\hat{u}_{\theta}$, as $\hat{u}_{\theta}$
are the residuals but at the derivative scale and not at the state
scale \cite{Hooker2009,HookerEllner2013}. 

Based on these remarks, we introduce the pertubed linear ODE 
\begin{equation}
\dot{x}(t)=A_{\theta}(t)x(t)+r_{\theta}(t)+u(t)\label{eq:ControlLinearODEmodel-1}
\end{equation}
where the function $t\mapsto u(t)$ can be any function in $L^{2}$.
The solution of the corresponding Initial Value Problem 
\[
\begin{array}{l}
\dot{x}(t)=A_{\theta}(t)x(t)+r_{\theta}(t)+u(t)\\
x(0)=X_{0}
\end{array}
\]
is denoted $X_{\theta,u}$. Instead of using the spline proxy $\hat{X}(\cdot,\theta)$
for approximating $X^{*}$, we use the trajectories $X_{\theta,u}$
of the ODE (\ref{eq:ControlLinearODEmodel-1}) controlled by the function
$u$.

\subsection{The Tracking Estimator}

Following the Generalized Smoothing approach, we look for a candidate
$X_{\theta,u}$ that can minimize at the same time the discrepancy
with the data and the norm $\left\Vert u\right\Vert _{L^{2}}$. Moreover,
we replace the classical Sum of Squared Errors by a smoothed version $\int_{0}^{T}\left\Vert \hat{X}(t)-X_{\theta,u}(t)\right\Vert _{2}^{2}dt$
based on a nonparametric proxy $\hat{X}$.
Hence, we consider the subsequent cost function 
\begin{equation}
C\left(\hat{X};u,\theta,\lambda\right)=\int_{0}^{T}\left\Vert \hat{X}(t)-X_{\theta,u}(t)\right\Vert _{2}^{2}dt+\lambda\int_{0}^{T}\left\Vert u(t)\right\Vert _{2}^{2}dt\label{eq:tracking_cost_function}
\end{equation}
for a given $\lambda>0$. Moreover, for each $\theta$ in $\Theta$,
we introduce the infimum function 
\begin{equation}
S\left(\hat{X};\theta,\lambda\right)=\inf_{u\in L^{2}}C\left(\hat{X};u,\theta,\lambda\right)\label{eq:ProfiledCost}
\end{equation}
obtained by ``profiling'' on the function $u$. The definition of
$S$ mimicks the minimization of $J_{n}(\beta\vert\theta,\lambda)$
but it is more involved as it is defined on infinite functional space,
instead of a finite dimensional vector space. Finally, our estimator
is defined by minimizing the same function $S$ i.e 

\begin{equation}
\widehat{\theta}^{T}=\arg\min_{\theta\in\Theta}S\left(\hat{X};\theta,\lambda\right)\label{eq:OurEstimator}
\end{equation}
whereas the GS estimator minimizes a different criterion $\sum_{i=1}^{n}\left\Vert Y_{i}-\hat{X}(t_{i},\theta)\right\Vert ^{2}$.
This means that in our methodology, we try to find a parameter $\theta$
that maintain a reasonable trade-off between model and data, whereas
the Generalized Smoothing Estimator $\hat{\theta}^{GS}$ is dedicated
to fit the data with the proxy $\hat{X}(\cdot,\theta)$, without considering
the size of model error represented by $\hat{u}_{\theta}$. 

Before going deeper into the interpretation and analysis of our estimator,
we need to show that the criterion $S\left(\hat{X};\theta,\lambda\right)$
is properly defined and that we can obtain a tractable expression
for computations and for the theoretical analysis of (\ref{eq:OurEstimator}).
The existence of $S$ is a direct consequence of the so-called Linear-Quadratic
Theory (LQ Theory), which belongs to the broader field of Optimal
Control Theory \cite{pontryagin1962minimum,Sontag1998,milyutin1998,clarke2013variationalcalculus}.
In our case, we consider the control of linear ODE with a quadratic
cost function that enables to have quite general and simple results.
This is possible because we have replaced the discrete sum of squared
errors by an integral criterion where the original data have been
replaced by a nonparametric proxy $\hat{X}$. Thanks to that, we can use
directly calculus of variations and optimal control \cite{Kirk1998optimalcontrol,clarke2013variationalcalculus}.
For completeness, we recall briefly in the appendix the main results
of LQ Theory.
\begin{namedthm}[Theorem and Definition of $S\left(\zeta;\theta,\lambda\right)$]
\label{thm:tracking_existence_unicity}Let $t\mapsto\zeta(t)$ be
a function belonging to $H^{1}(\left[0,\, T\right],\mathbb{R}^{d})$
and $X_{\theta,u}$ be the solution to the controlled ODE (\ref{eq:ControlLinearODEmodel-1}).
\\
For any $\theta$,$\lambda$, there exists a unique optimal control
$\bar{u}_{\theta,\lambda}$ that minimizes the cost function 
\begin{equation}
C\left(\zeta;u,\theta,\lambda\right)=\int_{0}^{T}\left\{ \left\Vert \zeta(t)-X_{\theta,u}(t)\right\Vert _{2}^{2}+\lambda\left\Vert u(t)\right\Vert _{2}^{2}\right\} dt\label{eq:LQCost}
\end{equation}
The control $\bar{u}_{\theta,\lambda}$ can be computed in a ``closed-loop''
form as 
\begin{equation}
\overline{u}_{\theta,\lambda}(t)=\frac{E(t)}{\lambda}\left(X_{\theta,\overline{u}_{\theta,\lambda}}(t)-\zeta(t)\right)+\frac{h(t)}{\lambda}\label{eq:OptimalControl}
\end{equation}
where $E$ and $h$ are solutions of the Final Value Problems 
\begin{equation}
\left\{ \begin{array}{l}
\dot{E}(t)=I_{d}-A_{\theta}(t)^{T}E(t)-E(t)A_{\theta}(t)-\frac{E(t)^{2}}{\lambda}\\
\dot{h}(t)=\text{\textminus}A_{\theta}(t)^{T}h(t)-E(t)\left(A_{\theta}(t)\zeta(t)+r_{\theta}(t)-\dot{\zeta}(t)\right)\text{\textminus}\frac{E(t)h(t)}{\lambda}
\end{array}\right.\label{eq:Ricatti_equ_fondamental}
\end{equation}
and $E(T)=0$, $h(T)=0$. For all $t\in[0,T]$, the matrix $E(t)$
is symetric, and the ODE defining the matrix-valued function $t\mapsto E(t)$
is called the Matrix Riccati Differential Equation of the ODE\textup{
(\ref{eq:ControlLinearODEmodel-1}). }

Finally, the Profiled Cost $S$ has the closed form \textup{
\begin{equation}
\begin{array}{lll}
S(\zeta;\theta,\lambda) & =& -\int_{0}^{T}\left\{ 2\left(A_{\theta}(t)\zeta(t)+r_{\theta}(t)-\dot{\zeta}(t)\right)^{\top}h(t)+\frac{\left\Vert h(t)\right\Vert ^{2}}{\lambda}\right\} dt\end{array}\label{eq:ExpressionS}
\end{equation}
}
\end{namedthm}
The cost (\ref{eq:LQCost}) is usually used for solving the so-called
``Tracking Problem'' that consists in finding the optimal control
$u$ to apply to the ODE (\ref{eq:ControlLinearODEmodel-1}) in order
to reach a target trajectory $t\mapsto\zeta(t)$, see \cite{Sontag1998}
for an excellent introduction. The estimation problem is then to determine
the parameter $\theta$ so that the corresponding ODE need a small
control $u$ (in $L^{2}$ norm) in order to be close to the noisy trajectory
$t\mapsto\hat{X}(t)$. 
\begin{rem}
$H^{1}(\left[0,\, T\right],\mathbb{R}^{d})=\left\{ X\in L^{2}(\left[0,\, T\right],\mathbb{R}^{d})\mid\dot{X}\in L^{2}(\left[0,\, T\right],\mathbb{R}^{d})\right\} $
\foreignlanguage{american}{is the classical Sobolev space of $L^2$ (weakly) differentiable function, see \cite{Brezis1983}.
The derivative is defined in the weak sense, so it allows us to consider non-parametric
estimator with some (controlled) discontinuities. Of course, every
differentiable functions belong to $H^{1}$ and the weak derivative
 coincides with the classic one}.
% where the weak derivative of
% the function $g$ in $H^{1}$ is not defined point-wise but as the
% function $\dot{g}\in L^{2}$ satisfying $\left\langle \dot{g},\varphi\right\rangle =-\left\langle g,\dot{\varphi}\right\rangle $,
% for all function $\varphi$ in $C^{1}$ with support included in $\left]0,1\right[$
%(denoted $C_{C}^{1}\left(\left]0,1\right[\right)$) . }
\end{rem}

\begin{rem}
We insist on the fact that $t\mapsto E(t),h(t)$ depends also on $\theta$,
$\lambda$ and $\zeta$ because of their definition via equation (\ref{eq:Ricatti_equ_fondamental}).
Nevertheless, we do not write it systematically for notational brievety.
As mentioned in the theorem, it is possible to compute $X_{\theta,\overline{u}_{\theta,\lambda}}$
in a ``closed-loop'' form as we can solve in a preliminary stage
the 2 equations (\ref{eq:Ricatti_equ_fondamental}) that gives the
function $E$ and $h$ for all $t\in\left[0,T\right]$. Then, we just
need to solve the ODE 
\[
\begin{array}{l}
\dot{x}(t)=A_{\theta}(t)x(t)+r_{\theta}(t)+\frac{E(t)}{\lambda}\left(x(t)-\zeta(t)\right)+\frac{h(t)}{\lambda}\\
x(0)=X_{0}
\end{array}
\]

\end{rem}

\begin{rem}
From equation (\ref{eq:ExpressionS}), we see that $S$ depends smoothly
in $\theta$ and $\lambda$, as in $\zeta$. This was not easy to
see from the infimum definition (\ref{eq:ProfiledCost}), but as the
minimum is reached, and attained for a known function, we can have
even more information than in the Generalized Smoothing approach based
on splines. 
\end{rem}

\begin{rem}
Our pertubated ODE framework permits to consider naturally the problem
of model misspecification, when the true model is 
\[
\dot{x}(t)=A_{\theta}(t)x(t)+r_{\theta}(t)+v(t)
\]
with $v\in L^{2}(\left[0,\, T\right],\mathbb{R}^{d})$ is an unknown
function. We do not provide any theoretical analysis for this kind
of model misspecification, but we consider it in the Experiments section,
in order to gain some insight. We will see in a simple example that
our estimator allows us to have more accurate estimation than classical
NLS estimator in that case. Moreover we can propose a proper correction
term to add to the initial model to counteract the misspecification
in order to lower the prediction error.
\end{rem}
The next section is dedicated to the derivation of the regularity
properties of $S$. Thanks to the use of a functional formulation
and the associated Linear-Quadratic theory, we show the smoothness in $\zeta$
and $\theta$, and compute directly the needed derivatives.

\section{Consistency of the Tracking estimator\label{sec:Consistency-of-the}}

Under reasonable and practical assumptions, we can assert
that the tracking estimator (\ref{eq:OurEstimator}) is a consistent
estimator of $\theta^{*}$ when the ODE model (\ref{eq:LinearODEmodel})
is well-specified, and when we use a consistent nonparametric estimator
$\hat{X}$. In practice, it is quite common to use a smoothing spline
or a kernel smoother in order to smooth the data and estimates roughly
the trajectory $X^{*}$. As the tracking estimator is an M-estimator,
we can employ the classical approaches for consistency that relies
on the regularity and convergence of the stochastic criterion $S\left(\hat{X};\theta,\lambda\right)$
to the asymptotic criterion $S\left(X^{*};\theta,\lambda\right)$.
Hence, we need to show some regularity in $\zeta$, uniformly in $\theta$.
Similarly, in order to compute the rate of convergence and the variance
of the estimator, we will need to check the smoothness w.r.t $\theta$.

\subsection{Regularity properties of $S(\zeta;\theta,\lambda)$}

We introduce some necessary assumptions about the ODE model in order
to derive the needed regularity as well as the identifiability property.
 The conditions are 
\begin{description}
\item [{C1:}] $\theta^{*}\in\Theta$ a compact subset of $\mathbb{R}^{p}$
\item [{C2:}] The model is identifiable at $\theta=\theta^{*}$ i.e 
\[
\forall\theta\in\Theta\,;\, X_{\theta}=X_{\theta^{*}}\Longrightarrow\theta=\theta^{*}
\]

\item [{C3:}] $\forall\left(t,\theta\right)\in\left[0\,,\, T\right]\times\Theta,\:(t,\theta)\longmapsto A_{\theta}(t)$
and $(t,\theta)\longmapsto r_{\theta}(t)$ are continuous.
\item [{C4:}] $\forall\left(t,\theta\right)\in\left[0,\, T\right]\times\Theta,\,\left(t,\theta\right)\longmapsto\frac{\partial A_{\theta}(t)}{\partial\theta}$
and $\left(t,\theta\right)\longmapsto\frac{\partial r_{\theta}(t)}{\partial\theta}$
are continuous
\end{description}
According to the context, $\left\Vert \cdot\right\Vert _{2}$ denotes
the Euclidean norm in $\mathbb{R}^{d}$ ($\left\Vert X\right\Vert _{2}=\sqrt{\sum_{i=1}^{d}X_{i}^{2}}$)
or the Frobenius matrix norm ($\left\Vert A\right\Vert _{2}=\sqrt{\sum_{i,j}\left|a_{i,j}\right|^{2}}$).
We use also the functional norm in $L^{2}\left(\left[0\, T\right],\mathbb{R}^{d}\right)$
defined by $\left\Vert f\right\Vert _{L^{2}}=\sqrt{\int_{0}^{T}\left\Vert f(t)\right\Vert _{2}^{2}dt}$.
Continuity and differentiability have to be understood w.r.t these
previous norms.

For the computation of $S\left(\hat{X};\theta,\lambda\right)$ (and
$S\left(X^{*};\theta,\lambda\right)$), we need some additional notations.
In particular, we recall that the Riccati equation $\dot{E}=I_{d}\text{\textminus}A_{\theta}(t)^{\top}E\text{\textminus}EA_{\theta}(t)\text{\textminus}\frac{E^{2}}{\lambda}$
depends on the model (\ref{eq:LinearODEmodel}) , but it does not
depend on the data $\hat{X}$, whereas it is the case for $h$, as we
have $\dot{h}(t)=\text{\textminus}A_{\theta}(t)^{T}h(t)-E(t)\left(A_{\theta}(t)\zeta(t)+r_{\theta}(t)-\dot{\zeta}(t)\right)\text{\textminus}\frac{E(t)h(t)}{\lambda}$.
For this reason, we introduce the functions $\alpha$ and $\beta$
defined by 
\[
\left\{ \begin{array}{l}
\alpha_{\theta}(t)=\left(A_{\theta}(t)^{T}+\frac{E_{\theta}(t)}{\lambda}\right)\\
\beta_{\theta}(t,\zeta)=E_{\theta}(t)\left(A_{\theta}(t)\zeta+r_{\theta}(t)-\dot{\zeta}\right)
\end{array}\right.
\]
We denote then $\widehat{h_{\theta}}$ the solution to the Final Value
Problem 
\[
\begin{cases}
\dot{h}=-\alpha_{\theta}(t)h-\beta_{\theta}(t,\widehat{X})\\
h(T)=0
\end{cases}
\]
and $h^{*}$ the solution corresponding to case $\zeta=X^{*}$. More
generally, we will denote $t\mapsto h_{\theta}(t,\zeta)$ for any
target trajectory $\zeta$. 

We introduce also the matrix-valued function $(t,s)\mapsto R_{\theta}(t,s)$
defined for all $t,s$ in $[0,T]$, as the solution of the Initial
Value Problem 
\begin{equation}
\textrm{ }\left\{ \begin{array}{l}
\dot{R_{\theta}}(t,s)=\alpha_{\theta}(T-t)R(t,s)\\
R_{\theta}(s,s)=I_{d}
\end{array}\right.\label{eq:reverse_time_R}
\end{equation}
and where the time has been reversed in the function $\alpha_{\theta}$.
 We show in the next proposition that $\forall\zeta\in H^{1}(\left[0,\, T\right])$,
$\theta\mapsto S(\zeta;\theta,\lambda)$ is well defined, i.e finite
on $\Theta$. 
\begin{prop}
\label{prop:existence_A_E_h}Under conditions 1 and 3 we have: 

\[
\overline{A}=\sup_{\theta\in\Theta}\left\Vert A_{\theta}\right\Vert _{L^{2}}<+\infty
\]
 
\[
\overline{X}=\sup_{\theta\in\Theta}\left\Vert X_{\theta}\right\Vert _{L^{2}}<+\infty
\]
 
\[
\bar{E}=\sup_{\theta\in\Theta}\left\Vert E_{\theta}\right\Vert _{L^{2}}<+\infty
\]
and 
\[
\forall\zeta\in H^{1}(\left[0,\, T\right]),\,\bar{h_{\zeta}}=\sup_{\theta\in\Theta}\left\Vert h_{\theta}(.,\zeta)\right\Vert _{L^{2}}<+\infty
\]
Hence, for all $\zeta$ in $H^{1}(\left[0,\, T\right])$, the map
$\theta\longmapsto S(\zeta;\theta,\lambda)$ is well defined on $\Theta$
(i.e $\sup_{\theta\in\Theta}\left\Vert S(\zeta;\theta,\lambda)\right\Vert <+\infty$). \end{prop}
\begin{proof}
$\bar{A}<+\infty$ exists as $(t,\theta)\mapsto A_{\theta}(t)$ is
a continuous function on the compact set $[0,T]\times\Theta$. The existence
and extension theorem for IVP solution of linear ODE ensures that
$\forall\theta\in\Theta,\,\left\Vert X_{\theta}\right\Vert _{L^{2}}<+\infty$.
Moreover, solutions are continuous in $(t,\theta)$ if the vector
field is continuous in $(t,\theta)$. By analogy with theorem \ref{thm:general_lq_theorem_existence_unicity} in appendix,
we know that 
\[
E_{\theta}^{g}(t)\,:=\left(\begin{array}{cc}
E_{\theta}(.) & h_{\theta}(.,\zeta)^{T}\\
h_{\theta}(.,\zeta) & \alpha_{\theta}(.,\zeta)
\end{array}\right)
\]
with $\alpha_{\theta}(t,\zeta)=\int_{t}^{T}\left(2\left(A_{\theta}(s)\zeta(s)-\dot{\zeta}(s)+r_{\theta}(s)\right)^{T}h_{\theta}(s,\zeta)+\frac{1}{\lambda}h_{\theta}(s,\zeta){}^{T}h_{\theta}(s,\zeta)\right)ds$
is the ODE solution of the extended Riccati ODE 
\[
\begin{array}{l}
\dot{E_{\theta}^{g}}(t)=W^{1}-A_{\theta}^{1}(t)^{t}E_{\theta}^{g}(t)-E_{\theta}^{g}(t)A_{\theta}^{1}(t)-\frac{1}{\lambda}E_{\theta}^{g}(t)^{2}\\
E_{\theta}^{g}(T)=0_{d+1,d+1}
\end{array}
\]
where $W_{1}=\left(\begin{array}{cc}
I_{d} & 0\\
0 & 0
\end{array}\right)$, $A_{\theta}^{1}(t)=\left(\begin{array}{cc}
A_{\theta}(t) & r_{\theta}^{1}(t)\\
0 & 0
\end{array}\right)$ and $r_{\theta}^{1}(t)=A_{\theta}(t)\zeta(t)+r_{\theta}(t)-\dot{\zeta}(t)$. 

Because for all $\theta\in\Theta,$ $A_{\theta}\in L^{2}\left(\left[0,\, T\right],\mathbb{R}^{d\times d}\right)$
and $(A_{\theta}\zeta-\dot{\zeta}+r_{\theta})\in L^{2}\left(\left[0,\, T\right],\mathbb{R}^{d}\right)$
thanks to \lemref{Ricatti_existence} in appendix, $E_{\theta}^{g}(t)$
is bounded and continuous in $(t,\theta)$. Hence $h_{\theta},E_{\theta}$
are bounded on $[0,T]\times\Theta$. We conclude for $\theta\mapsto S(\zeta;\theta,\lambda)$
thanks to norm inequality.
\end{proof}
We complete our analysis by showing that $S$ is continuously differentiable in $\theta$. 
\begin{prop}
\label{prop:c1_caracterisation_S}Under conditions C1-C3 
\[
\forall X\in H^{1}(\left[0,\, T\right]),\,\theta\longmapsto S(X;\theta,\lambda)
\]
is continuous on $\Theta$. \\
Under conditions C1-C4, $S$ is $C^{1}$ on $\Theta$.\end{prop}
\begin{proof}
Since
\[
S(X;\theta,\lambda)=-\int_{0}^{T}\left(2\left(A_{\theta}(t)X(t)+r_{\theta}(t)-\dot{X}(t)\right)^{T}h_{\theta}(t,X)+\frac{1}{\lambda}\left\Vert h_{\theta}(t,X)\right\Vert ^{2}\right)dt
\]
Condition 3, jointly with proposition 1 and 4 in the supplementary
materials give the continuity of $\theta\longmapsto\left(t\longmapsto A_{\theta}(t)\right)$
and $\left(\theta,X\right)\longmapsto\left(t\longmapsto h_{\theta}(t,X)\right)$
on $\Theta$. This is enough to show the continuity of $\theta\longmapsto S(X;\theta,\lambda)$
on $\Theta$. Moreover, the gradient w.r.t $\theta$ of $S(X;\theta,\lambda)$
is equal to:

\[
\begin{array}{lll}
\nabla_{\theta}S(X;\theta,\lambda) & = & -2\int_{0}^{T}\frac{\partial\left(A_{\theta}(t).X+r_{\theta}(t)\right)}{\partial\theta}^{T}h_{\theta}(t,X)dt\\
 &  & +2\int_{0}^{T}\frac{\partial h_{\theta}(t,X)}{\partial\theta}^{T}\left(A_{\theta}(t).X+r_{\theta}(t)-\dot{X}+\frac{1}{\lambda}h_{\theta}(t,X)\right)dt
\end{array}
\]
In addition to the previous proposition, condition 4 and proposition
7 in supplementary material gives the continuity of $\left(\theta,X\right)\longmapsto\left(t\longmapsto\frac{\partial\left(h_{\theta}(t,X)\right)}{\partial\theta}\right)$
on $\Theta$. This is enough to show the continuous differentiability
of $S(X;\theta,\lambda)$ on $\Theta$. 
\end{proof}
The last regularity properties justifies the use of classical optimization
method to retrieve the minimum of $S$.

In the next proposition, we show that the criteria $S(X;\theta,\lambda)$ can be expressed without using the derivative $\dot{X}$ (thanks to the knowledge of the initial condition). As a consequence, our estimator is less sensible to the nonparametric noise than classical Two-Step estimators.
\begin{prop}
Under conditions 1 and 2, $\forall X\in H^{1}(\left[0,\, T\right])$
with $X(0)=X_{0}^{*}$, $S(X;\theta,\lambda)$ does not depend on
$\dot{X}$, i.e it is continuous nonlinear integral of $t\mapsto X(t)$. \end{prop}
\begin{proof}
We will show $S(X;\theta,\lambda)$ can be written using only $X$
and not $\dot{X}$. First of all we will use  Lemma (\ref{lem:contolled_X_der_h_by_X_h})
to get rid of $\dot{X}$ in $\int_{0}^{T}\dot{X}(t)^{T}h_{\theta}(t,X)dt$,
it gives: 
\begin{equation}
\begin{array}{lll}
\int_{0}^{T}\dot{X}(t)^{T}h_{\theta}(t,X)dt & = & F_{1,\theta}(X)+F_{2,\theta}(X)+F_{3,\theta}(X)\\
 &  & -X_{0}^{*T}\int_{0}^{T}R_{\theta}(T,T-s)E_{\theta}(s)r_{\theta}(s)ds\\
 &  & -\frac{1}{2}X_{0}^{*T}E_{\theta}(0)X_{0}^{*}
\end{array}\label{eq:_X_der_without_der}
\end{equation}
with

\[
\left\{ \begin{array}{l}
F_{1,\theta}(X)=-X_{0}^{T}\int_{0}^{T}R_{\theta}(T,T-s)X(s)ds\\
F_{2,\theta}(X)=\int_{0}^{T}X(t)^{T}\left(\alpha_{\theta}(t)h_{\theta}(t,X)dt+\left(A_{\theta}(t)X(t)+r_{\theta}(t)\right)\right)dt\\
F_{3,\theta}(X)=\frac{1}{2}\int_{0}^{T}X(t)^{T}\dot{E_{\theta}(t)}X(t)dt
\end{array}\right.
\]
And so we can write $S(X;\theta,\lambda)$ under the form

\[
\begin{array}{lll}
S(X;\theta,\lambda) & = & -\int_{0}^{T}\left(2\left(A_{\theta}(t)X(t)+r_{\theta}(t)\right)^{T}h_{\theta}(t,X)+\frac{1}{\lambda}h_{\theta}(t,X)^{T}h_{\theta}(t,X)\right)dt\\
 &  & +F_{1,\theta}(X)+F_{2,\theta}(X)+F_{3,\theta}(X)\\
 &  & -X_{0}^{*T}\int_{0}^{T}R_{\theta}(T,T-s)E_{\theta}(s)r_{\theta}(s)ds\\
 &  & -\frac{1}{2}X_{0}^{*T}E_{\theta}(0)X_{0}^{*}
\end{array}
\]
since from Lemma (\ref{lem:linear_h_wrt_X}) we have the affine dependence
of $h_{\theta}(t,X)$ w.r.t $X$ through the formula: 
\[
h_{\theta}(t,X)=\int_{t}^{T}R_{\theta}(T-t,T-s)X(s)ds+E_{\theta}(t)X(t)+\int_{t}^{T}R_{\theta}(T-t,T-s)E_{\theta}(s)r_{\theta}(s)ds
\]
we see $S(X;\theta,\lambda)$ does not depend on $\dot{X}$.
\end{proof}

\subsection{Consistency}

As we have seen previously, conditions 1 and 3 ensure the existence
of $S(\hat{X};\theta,\lambda)$ and $S(X^{*};\theta,\lambda)$ for
all $\theta\in\Theta$. We derive the consistency of $\hat{\theta}^{T}$
by showing the uniform convergence of the criterion $S\left(\hat{X};\theta,\lambda\right)$,
and by insuring that $\theta^{*}$ is a unique and isolated global
minima of $S\left(X^{*};\theta,\lambda\right)$. Condition 2 is then
sufficient to show that $S\left(X^{*};\theta,\lambda\right)$ characterizes
well $\theta^{*}$, as global unique minimum.  Hence,  identifiability
and convergence in supremum norm are sufficient to imply consistency
(theorem 5.7 in \cite{Vaart1998}).
\begin{prop}
\label{prop:asym_identifiability}For all $X$ in $H^{1}(\left[0,\, T\right])$,
$S(X;\theta,\lambda)\geq0$ and under conditions C1 and C2 we have
\[
S(X^{*};\theta,\lambda)=0\Longleftrightarrow\theta=\theta^{*}
\]
\end{prop}
\begin{proof}
If $\theta=\theta^{*}$, then $u\equiv0$ is the cost which minimizes
\[
C\left(X^{*};u,\theta^{*},\lambda\right)=\int_{0}^{T}\left\Vert X^{*}(t)-X_{\theta^{*},u}(t)\right\Vert _{2}^{2}dt+\lambda\int_{0}^{T}\left\Vert u(t)\right\Vert _{2}^{2}dt
\]
 and in that case $S(X^{*};\theta^{*},\lambda)=\inf_{u\in L^{2}}C\left(X^{*};u,\theta^{*},\lambda\right)=0$. 

Conversely, let $\theta^{0}$ be such that $S(X^{*};\theta^{0},\lambda)=0$.
By definition, this means that $\int_{0}^{T}\left\Vert X^{*}(t)-X_{\theta^{0},u}(t)\right\Vert _{2}^{2}dt+\lambda\int_{0}^{T}\left\Vert u(t)\right\Vert _{2}^{2}dt=0$.
A consequence is that $u=0\, a.e$ and $X_{\theta^{*}u=0}(t)=X_{\theta^{0},u=0}(t)\, a.e$;
by the identifiability condition we get that $\theta^{0}=\theta^{*}$. \end{proof}
\begin{thm}
\label{thm:consistency}Under conditions 1, 2, 3 and if $\widehat{X}$
is consistent in probability (in $L^{2}-$norm sense), we have 
\[
\widehat{\theta}^{T}\overset{P}{\rightarrow}\theta^{*}
\]
\end{thm}
\begin{proof}
Using proposition \propref{controlled_S_by_h_X}, we have

\[
\begin{array}{l}
\left|S(X;\theta,\lambda)-S(X^{*};\theta,\lambda)\right|\\
\leq2\left(\bar{A}\bar{h}+K_{1}+K_{2}\left\Vert \widehat{h}_{\theta}\right\Vert _{L^{2}}+K_{3}\left\Vert \widehat{X}\right\Vert _{L^{2}}\right)\left\Vert X^{*}-\widehat{X}\right\Vert _{L^{2}}\\
+\left(\bar{A}\left\Vert \widehat{X}\right\Vert _{L^{2}}+K_{4}+\frac{1}{\lambda}\left(\left\Vert \widehat{h}_{\theta}\right\Vert _{L^{2}}+\bar{h}\right)\right)\left\Vert h_{\theta}^{*}-\widehat{h}_{\theta}\right\Vert _{L^{2}}
\end{array}
\]
with 
\[
\begin{array}{l}
K_{1}=\sqrt{d}\left\Vert X_{0}\right\Vert _{2}\bar{R}+\sqrt{d}\bar{A}\bar{X}+\sqrt{d}\bar{\dot{E}}\overline{X}\\
K_{2}=\sqrt{d}\left(\bar{A}+\frac{\bar{E}}{\lambda}\right)\\
K_{3}=\sqrt{d}\bar{A}+\sqrt{d}\bar{\dot{E}}\\
K_{4}=\sqrt{d}\left(\bar{A}+\frac{\bar{E}}{\lambda}\right)\bar{X}
\end{array}
\]
and 
\[
\begin{array}{l}
\bar{R}=\sup_{\theta\in\Theta}\left\Vert R_{\theta}(T,T-.)\right\Vert _{L^{2}}\\
\bar{\dot{E}}=\sup_{\theta\in\Theta}\left\Vert \dot{E}_{\theta}\right\Vert _{L^{2}}
\end{array}
\]
by using the same notation as in proposition \ref{prop:existence_A_E_h}.
Proposition \propref{controlled_h_by_X} permits to bound $\left\Vert h_{\theta}^{*}-\widehat{h}_{\theta}\right\Vert _{L^{2}}$
with $\left\Vert \widehat{X}-X^{*}\right\Vert _{L^{2}}$ as 
\[
\begin{array}{l}
\left\Vert \widehat{h_{\theta}}-h_{\theta}^{*}\right\Vert _{L^{2}}\leq K_{5}\left\Vert \widehat{X}-X^{*}\right\Vert _{L^{2}}\\
\textrm{with : }K_{5}=\sqrt{d}\left(Tde^{\sqrt{d}\left(\overline{A}+\frac{\overline{E}}{\lambda}\right)T}+\overline{E}\right)
\end{array}
\]
We obtain
\[
\begin{array}{l}
\begin{array}{l}
\left|S(X;\theta,\lambda)-S(X^{*};\theta,\lambda)\right|\leq\left(\left(2K_{2}+\frac{K_{5}}{\lambda}\right)\left\Vert \widehat{h_{\theta}}\right\Vert _{L^{2}}+\left(2K_{3}+K_{5}\bar{A}\right)\left\Vert \widehat{X}\right\Vert _{L^{2}}+K_{7}\right)\left\Vert X^{*}-\widehat{X}\right\Vert _{L^{2}}\\
\textrm{with: }K_{7}=2\left(\bar{A}\bar{h}+K_{1}\right)+K_{5}\left(K_{4}+\frac{\bar{h}}{\lambda}\right)
\end{array}\end{array}
\]
We can control $\left\Vert \widehat{X}\right\Vert _{L^{2}}\leq\left\Vert \widehat{X}-X^{*}\right\Vert _{L^{2}}+\left\Vert X^{*}\right\Vert _{L^{2}}$,
which proves that if is $\widehat{X}$ is consistent, then $\sup_{\theta\in\Theta}\left|S(X;\theta,\lambda)-S(X^{*};\theta,\lambda)\right|=o_{P}(1)$.
Application of the proposition \propref{asym_identifiability} gives
us the identifiability criteria. Hence we conclude by using the theorem
5.7 in \cite{Vaart1998}.
\end{proof}

\section{Asymptotics of $\widehat{\theta}^{T}$\label{sec:Asymptotics-of}}

Our objective in this part is to derive the proper rate of convergence
of the Tracking Estimator, as well as its asymptotic distribution.
The properties of the estimator depends on the behavior of the nonparametric
estimate $\hat{X}$ used for the approximation of $X^{*}$. In order
to fix ideas, we consider a regression spline, with a B-Spline decomposition
of dimension $K$ (increasing with $n$). That is we consider that
$\widehat{X}$ is defined as
\[
\widehat{X}(t)=\sum_{k=1}^{K}\beta_{kK}p_{kK}(t)=\beta_{K}^{T}p_{K}(t)
\]
where $\beta_{K}$ is computed by least-squares. It is likely that
we could derive the same kind of results for different estimates,
such as Local Polynomial or Smoothing Splines, as they behave similarly
asymptotically, and that we show that the Tracking Estimate can be
approximated by a plug-in estimate of a specific linear functional
of $\hat{X}$. We introduce additional regularity conditions needed
for the asymptotics: 
\begin{description}
\item [{C5:}] The Hessian $\frac{\partial^{2}S(X^{*};\theta,\lambda)}{\partial\theta^{T}\partial\theta}$
is nonsingular at $\theta=\theta^{*}$. 
\item [{C6:}] The observations $\left(t_{i},Y_{i}\right)$ are i.i.d with
$Var(Y_{i}\mid t_{i})=\sigma I_{d}$ with $\sigma<+\infty$
\item [{C7:}] Observations time $t_{i}$ are uniformely distributed on
$\left[0\,,\, T\right]$
\item [{C8:}] It exists $s\geq1$ such that $t\longmapsto A_{\theta^{*}}(t)$,
$t\longmapsto r_{\theta^{*}}(t)$ are $C^{s-1}\left(\left[0\,,\, T\right],\mathbb{R}^{d}\right)$
and $\sqrt{n}K^{-s}\longrightarrow0$ and $\frac{K^{s}}{n}\longrightarrow0$
\end{description}
Under these additional conditions, we show that $\hat{\theta}^{T}$
reaches the parametric convergence rate, and that it is asymptotically
normal. Our strategy consists in two stages:
\begin{lyxlist}{00.00.0000}
\item [{Stage~1~(Proposition~\ref{prop:linear_decomposition})}] We show that
$\hat{\theta}^{T}-\theta^{*}$ behaves asymptotically as the difference
$\Gamma(\widehat{X})-\Gamma(X^{*})$ where $\Gamma$ is a continuous
linear functional,

\item [{Stage~2~(Proposition~\ref{prop:newey_result})}] We
prove that $\Gamma\left(\widehat{X}-X^{*}\right)$ is asymptotically
normal for regression splines, based on the properties of plug-in
estimators computed with series estimators and derived in \cite{Newey1997}.
 \end{lyxlist}

\begin{rem}
Condition C5 is a classic feature for $M-$estimator to ensure local
identifiability, here: 
\[
\begin{array}{lll}
\frac{\partial^{2}S(X^{*};\theta^{*},\lambda)}{\partial\theta^{T}\partial\theta} & = & 2\int_{0}^{T}\frac{\partial\left(A_{\theta^{*}}(t)X^{*}+r_{\theta^{*}}(t)\right)}{\partial\theta}^{T}\frac{\partial h_{\theta^{*}}^{*}(t)}{\partial\theta}+\frac{\partial h_{\theta^{*}}^{*}(t)}{\partial\theta}^{T}\frac{\partial\left(A_{\theta^{*}}(t)X^{*}+r_{\theta^{*}}(t)\right)}{\partial\theta}dt\\
 & + & \frac{2}{\lambda}\int_{0}^{T}\frac{\partial h_{\theta^{*}}^{*}(t)}{\partial\theta}^{T}\frac{\partial h_{\theta^{*}}^{*}(t)}{\partial\theta}dt
\end{array}
\]
that is why we only require $\forall\left(t,\theta\right)\in\left[0\,,\, T\right]\times\Theta,\:(t,\theta)\longmapsto A_{\theta}(t)$
and $(t,\theta)\longmapsto r_{\theta}(t)$ to be $C^{1}$ and not
$C^{2}$
\end{rem}

\begin{rem}
Condition C8 is a classic feature for non-parametric estimator to
ensure optimal convergence rate of $\widehat{X}$ using bias-variance
tradeoff.\end{rem}
\begin{prop}
\label{prop:linear_decomposition}Under conditions 1-5, we have :
\[
\widehat{\theta}^{T}-\theta^{*}=2\frac{\partial^{2}S(X^{*};\theta^{*},\lambda)}{\partial\theta^{T}\partial\theta}^{-1}\left(\Gamma(\widehat{X})-\Gamma(X^{*})\right)+o_{P}(1)
\]
where $\Gamma\::\: C\left(\left[0\,,\, T\right],\mathbb{R}^{d}\right)\:\rightarrow\mathbb{R}^{p}$
is a linear functional defined by 
\begin{equation}
\Gamma(X)=\int_{0}^{T}\left(\frac{\partial\left(A_{\theta*}(t).X^{*}\right)}{\partial\theta}+\frac{1}{\lambda}\frac{\partial h_{\theta^{*}}(t,X^{*})}{\partial\theta}\right)^{T}\left(\int_{t}^{T}R_{\theta^{*}}(T-t,T-s)X(s)ds\right)dt.\label{eq:DefinitionGamma}
\end{equation}
 $R_{\theta^{*}}$ is defined \textup{by (\ref{eq:reverse_time_R}).}\end{prop}

\begin{prop}
\label{prop:newey_result}Under conditions 1-8 and by defining $\Gamma$
as in proposition \propref{linear_decomposition} we have that $\Gamma(\widehat{X})-\Gamma(X^{*})$
is asymptotically normal and $\Gamma(\widehat{X})-\Gamma(X^{*})=O_{P}(n^{-1/2})$
\end{prop}
To obtain the final result, we only have to combine the two previous
propositions:
\begin{thm}
\label{thm:asymptotic}If $\widehat{X}$ is a regression spline and
conditions C1-C8 are satisfied, then $\widehat{\theta}^{T}-\theta^{*}$
is asymptotically normal and 
\[
\widehat{\theta}^{T}-\theta^{*}=O_{P}(n^{-1/2})
\]
\end{thm}
\begin{rem}
The asymptotic linear representation given by proposition \ref{prop:linear_decomposition}
allows us to obtain an expression for the asymptotic variance. This
latter is given by the formula (\ref{eq:as_variance_expression})
in appendix D as well as a plug-in consistent estimator.
\end{rem}

\section{Experiments\label{sec:Experiments}}

We use several simple test beds for evaluating the practical efficiency
of the Tracking Estimator $\hat{\theta}^{T}$, and we compare it with
the performance of NLS $\hat{\theta}^{NLS}$and of Generalized Smoothing
$\hat{\theta}^{GS}$. The different models are linear in the states,
and they can be linear and nonlinear w.r.t parameters. We use several
sample size and several variance error for comparing robustness and
efficiency for varying sample size and noise level.

\subsection{Experimental design}

For a given sample size $n$ and noise level $\sigma$, we estimate
the Mean Square Error and the mean Absolute Relative Error (ARE) 
\[
\mathbb{E}_{\theta^{\text{*}}}\left[\frac{\left|\theta^{*}-\widehat{\theta}\right|}{\left|\theta^{*}\right|}\right]
\]
by Monte Carlo, based on $N_{MC}=100$ runs. For each run, we simulate
an ODE solution with a Runge-Kutta algorithm (ode45 in Matlab), and
a centered Gaussian noise (with variance $\sigma$) is added, in order
to obtain the $Y_{i}$'s. 

We compare also the quality of estimation of  $\hat{\theta}^{T}$, $\widehat{\theta}^{GS}$
and $\widehat{\theta}^{NLS}$ based on their prediction quality. We compute the Prediction Error  
\begin{equation}
\mathrm{\mathbb{E}_{\theta^{\text{*}},\sigma}\left[\left\Vert Y^*-X_{\widehat{\theta}}\right\Vert _{L^{2}}^{2}\right]}\label{eq:prediction_error}
\end{equation}
where
\begin{itemize}
\item $Y^*$ is a new observation drawn from the true model (\ref{eq:StatModel}),
\item $X_{\widehat{\theta}}$ is the solution to the linear ODE (\ref{eq:LinearODEmodel}) with parameter $\widehat{\theta}$.
\end{itemize}
It should be emphasized that parameter estimation and prediction error
minimization are two different problems, although they are related. Parameter estimation
is required when parameters are directly of interest, for instance for a deep 
understanding of the inner dynamics of the system. One put forward
prediction error when the aim is only to quantitatively predict the
system state. Our primary interest is parameter estimation but we
also discuss prediction performance for the three methods.

The nonparametric estimate $\hat{X}$ required in the first step is a spline defined 
with a uniform knots sequence $\xi_{k},k=1,\dots,K$. For each run
and each state variables, the number of knots is selected by minimizing
the GCV criterion, \cite{ruppert2003semiparametric}.\\
 We discuss in the next section an automated method for selecting adaptively the
hyperparameter $\lambda$.

\subsection{Selection of $\lambda$}

The criterion $S\left(\hat{X};\theta,\lambda\right)$ is based on
a balance between data fidelity and model fidelity. When $\lambda\rightarrow0$, 
we can select any $u$ in order to interpolate $\hat{X}$. In that case, $\theta$ has almost no influence on $S\left(\hat{X};\theta,\lambda\right)$
value.\\ 
When $\lambda\rightarrow\infty$, the optimal perturbation $\overline{u}\longrightarrow0$, 
and we get a NLS-like criterion where the observations $Y_i$'s are replaced by the proxy $\widehat{X}$. Because of this dramatic influence, we propose to select  $\lambda$ by minimizing the Sum of Squared Errors 
\[SSE(\lambda) = \sum_{i=1}^n\left(Y_{i}-X_{\widehat{\theta_{\lambda}}^{T}}(t_{i})\right)^{2}\].
%or the corrected SSE: $\sum_{i}\left(Y_{i}-X_{\widehat{\theta_{\lambda}}^{T},u}(t_{i})\right)^{2}$, but due to the overfitting issue the latter drove use to choose the smallest $\lambda$ possible.

% So in order to choose $\lambda$ for each model we have tested a trial vector of smoothing value and retain the one minimizing $\sum_{i}\left(Y_{i}-X_{\widehat{\theta_{\lambda}}^{T}}(t_{i})\right)^{2}$.

%Our aim was to minimize on function $g(\lambda)$ such $\mathbb{E}_{\theta^{\text{*}}}\left[\left\Vert \theta^{*}-\widehat{\theta_{\lambda}}^{T}\right\Vert _{2}^{2}\right]\leq g(\lambda)$,
% because of the uneasiness of finding such functions 

\subsection{Gradient computation}

For optimization purpose, we need to compute the gradient $\nabla_{\theta}S(X;\theta,\lambda)$
which involves both $\frac{\partial\widehat{h_{\theta}}}{\partial\theta}$
and $\frac{\partial E_{\theta}}{\partial\theta}$. These partial derivatives can be obtained by solving the sensitivity equations,  that gives the function values
at each time $t\in\left[0,\, T\right]$. Nevertheless, the size of the ODE to solve grows quickly, as the sensitivity systems is of size $(d^{2}+d)\times p$, and it becomes a computational burden for the optimization
process. For this reason, we use the adjoint method to compute gradient
expression \cite{Cao2003}. This method exploits the fact that we do not need a point-wise computation of $\frac{\partial\widehat{h_{\theta}}}{\partial\theta}$
and $\frac{\partial E_{\theta}}{\partial\theta}$ but only some integral of the derivatives.
An explanation of this method for gradient and Hessian computation
can be found in \cite{Blayo2011}. The Riccati ODE can be written in vector form (of size $D:=d^2+d$)
\[
\begin{array}{l}
\dot{Q_{\theta}}=F(Q_{\theta},\theta,t)\\
Q_{\theta}(T)=0
\end{array}
\]
where $F$ is the row formulation of the Riccati ODE vector field and the solution is the vector \[
Q_{\theta}(t)=\left(\widehat{h_{\theta}}^{T},\left(E_{\theta}^{r}\right)^{T}\right)^{T}(t).
\] 
Hence, we can compute $\nabla_{\theta}S(X;\theta,\lambda)$ thanks to the
formula
\[
\begin{array}{c}
\nabla_{\theta}S(X;\theta,\lambda)= \int_{0}^{T} \{ \frac{\partial g(Q_{\theta}(t),\theta,t)}{\partial Q}-P(t).\frac{\partial F}{\partial\theta}(Q_{\theta}(t),\theta,t) \} dt\end{array}
\]
with
\[
g(Q_{\theta},\theta,t)=-2\left(A_{\theta}(t)\widehat{X}(t)-\dot{\widehat{X}}(t)+r_{\theta}(t)\right)^{T}\widehat{h_{\theta}}-\frac{1}{\lambda} \Vert \widehat{h_{\theta}}\Vert^2
\]
and $P$ is the so-called adjoint vector of size $D=d^{2}+d$, solution of the adjoint model 
\[
\begin{array}{l}
\begin{array}{l}
\dot{P}(t)=\frac{\partial g(Q_{\theta}(t),\theta,t)}{\partial Q}-P(t).\frac{\partial F}{\partial Q}(Q_{\theta},\theta,t)\\
P(0)=0
\end{array}\end{array}
\]
The computational details for $\frac{\partial g}{\partial\theta},\frac{\partial g}{\partial Q},\frac{\partial F}{\partial\theta},\frac{\partial F}{\partial Q}$
are left in appendix B. The adjoint method is more efficient than the direct sensitivity approach, as we need to solve an ODE of size $D$, instead of a system of size $D\times p$, which is valuable when the number of parameters increases.

\subsection{Simple scalar equation}

\subsubsection{Linear w.r.t parameter}

We consider here the basic model linear in parameter
\begin{equation}
\dot{x}=ax\label{eq:exp_scalar_ode1}
\end{equation}
with initial condition equal to $X_{0}^{*}=1$. It is the simplest
model we can consider, here the solution of the Cauchy problem is
$x(t)=X_{0}^{*}e^{at}$ and we will use this closed form for the NLS
estimator. Here we have tested two different sample size $n=50$ and
$n=20$ (observations were uniformly distributed between $0$ and $5$)
and two different noise level $\sigma=2$ and $\sigma=4$. The
lambda values tested are the $40$ values uniformly distributed between
$10$ and $400$. 

\begin{center}
{\small{}}
\begin{table}[h]
\selectlanguage{american}%
\begin{centering}
{\small{}}%
\begin{tabular}{|c|c|c|c|c|}
\hline 
\selectlanguage{english}%
{\small{}$\left(n,\sigma\right)$}\selectlanguage{american}%
 & \selectlanguage{english}%
\selectlanguage{american}%
 & \multicolumn{1}{c|}{\selectlanguage{english}%
{\small{}MSE ($\times10^{-5}$)}\selectlanguage{american}%
} & \selectlanguage{english}%
{\small{}ARE ($\times10^{-3}$)}\selectlanguage{american}%
 & \selectlanguage{english}%
{\small{}Pred error }\selectlanguage{american}%
\tabularnewline
\hline 
\hline 
\selectlanguage{english}%
{\small{}$\left(50,2\right)$}\selectlanguage{american}%
 & \selectlanguage{english}%
\selectlanguage{american}%
 & \selectlanguage{english}%
\selectlanguage{american}%
 & \selectlanguage{english}%
\selectlanguage{american}%
 & \selectlanguage{english}%
\selectlanguage{american}%
\tabularnewline
\hline 
\selectlanguage{english}%
\selectlanguage{american}%
 & \selectlanguage{english}%
{\small{}$\widehat{\theta}^{T}$}\selectlanguage{american}%
 & \selectlanguage{english}%
{\small{}1.19}\selectlanguage{american}%
 & \selectlanguage{english}%
{\small{}3.2}\selectlanguage{american}%
 & \selectlanguage{french}%
4.52\selectlanguage{american}%
\tabularnewline
\hline 
\selectlanguage{english}%
\selectlanguage{american}%
 & \selectlanguage{english}%
{\small{}$\widehat{\theta}^{NLS}$}\selectlanguage{american}%
 & \selectlanguage{english}%
{\small{}1.25}\selectlanguage{american}%
 & \selectlanguage{english}%
{\small{}3.5}\selectlanguage{american}%
 & \selectlanguage{french}%
4.53\selectlanguage{american}%
\tabularnewline
\hline 
\selectlanguage{english}%
\selectlanguage{american}%
 & \selectlanguage{english}%
{\small{}$\widehat{\theta}^{GS}$}\selectlanguage{american}%
 & \selectlanguage{french}%
43\selectlanguage{american}%
 & \selectlanguage{french}%
22.5\selectlanguage{american}%
 & \selectlanguage{french}%
6.06\selectlanguage{american}%
\tabularnewline
\hline 
\selectlanguage{english}%
{\small{}$\left(50,4\right)$}\selectlanguage{american}%
 & \selectlanguage{english}%
\selectlanguage{american}%
 & \selectlanguage{english}%
\selectlanguage{american}%
 & \selectlanguage{english}%
\selectlanguage{american}%
 & \selectlanguage{english}%
\selectlanguage{american}%
\tabularnewline
\hline 
\selectlanguage{english}%
\selectlanguage{american}%
 & \selectlanguage{english}%
{\small{}$\widehat{\theta}^{T}$}\selectlanguage{american}%
 & \selectlanguage{english}%
{\small{}3.42}\selectlanguage{american}%
 & \selectlanguage{english}%
{\small{}5.3}\selectlanguage{american}%
 & \selectlanguage{french}%
9.05\selectlanguage{american}%
\tabularnewline
\hline 
\selectlanguage{english}%
\selectlanguage{american}%
 & \selectlanguage{english}%
{\small{}$\widehat{\theta}^{NLS}$}\selectlanguage{american}%
 & \selectlanguage{english}%
{\small{}3.91}\selectlanguage{american}%
 & \selectlanguage{english}%
{\small{}6.0}\selectlanguage{american}%
 & \selectlanguage{french}%
9.07\selectlanguage{american}%
\tabularnewline
\hline 
\selectlanguage{english}%
\selectlanguage{american}%
 & \selectlanguage{english}%
{\small{}$\widehat{\theta}^{GS}$}\selectlanguage{american}%
 & \selectlanguage{french}%
230\selectlanguage{american}%
 & \selectlanguage{french}%
47.2\selectlanguage{american}%
 & \selectlanguage{french}%
11.76\selectlanguage{american}%
\tabularnewline
\hline 
\selectlanguage{english}%
{\small{}$\left(20,2\right)$}\selectlanguage{american}%
 & \selectlanguage{english}%
\selectlanguage{american}%
 & \selectlanguage{english}%
\selectlanguage{american}%
 & \selectlanguage{english}%
\selectlanguage{american}%
 & \selectlanguage{english}%
\selectlanguage{american}%
\tabularnewline
\hline 
\selectlanguage{english}%
\selectlanguage{american}%
 & \selectlanguage{english}%
{\small{}$\widehat{\theta}^{T}$}\selectlanguage{american}%
 & \selectlanguage{english}%
{\small{}2.56}\selectlanguage{american}%
 & \selectlanguage{english}%
{\small{}4.6}\selectlanguage{american}%
 & \selectlanguage{french}%
4.58\selectlanguage{american}%
\tabularnewline
\hline 
\selectlanguage{english}%
\selectlanguage{american}%
 & \selectlanguage{english}%
{\small{}$\widehat{\theta}^{NLS}$}\selectlanguage{american}%
 & \selectlanguage{english}%
{\small{}2.92}\selectlanguage{american}%
 & \selectlanguage{english}%
{\small{}5.6}\selectlanguage{american}%
 & \selectlanguage{french}%
4.74\selectlanguage{american}%
\tabularnewline
\hline 
\selectlanguage{english}%
\selectlanguage{american}%
 & \selectlanguage{english}%
{\small{}$\widehat{\theta}^{GS}$}\selectlanguage{american}%
 & \selectlanguage{french}%
100\selectlanguage{american}%
 & \selectlanguage{french}%
32.7\selectlanguage{american}%
 & \selectlanguage{french}%
7.48\selectlanguage{american}%
\tabularnewline
\hline 
\selectlanguage{english}%
{\small{}$\left(20,4\right)$}\selectlanguage{american}%
 & \selectlanguage{english}%
\selectlanguage{american}%
 & \selectlanguage{english}%
\selectlanguage{american}%
 & \selectlanguage{english}%
\selectlanguage{american}%
 & \selectlanguage{english}%
\selectlanguage{american}%
\tabularnewline
\hline 
\selectlanguage{english}%
\selectlanguage{american}%
 & \selectlanguage{english}%
{\small{}$\widehat{\theta}^{T}$}\selectlanguage{american}%
 & \selectlanguage{english}%
{\small{}8.82}\selectlanguage{american}%
 & \selectlanguage{english}%
{\small{}7.8}\selectlanguage{american}%
 & \selectlanguage{french}%
9.07\selectlanguage{american}%
\tabularnewline
\hline 
\selectlanguage{english}%
\selectlanguage{american}%
 & \selectlanguage{english}%
{\small{}$\widehat{\theta}^{NLS}$}\selectlanguage{american}%
 & \selectlanguage{english}%
{\small{}9.41}\selectlanguage{american}%
 & \selectlanguage{english}%
{\small{}9.3}\selectlanguage{american}%
 & \selectlanguage{french}%
9.10\selectlanguage{american}%
\tabularnewline
\hline 
\selectlanguage{english}%
\selectlanguage{american}%
 & \selectlanguage{english}%
{\small{}$\widehat{\theta}^{GS}$}\selectlanguage{american}%
 & \selectlanguage{french}%
440\selectlanguage{american}%
 & \selectlanguage{french}%
66.0\selectlanguage{american}%
 & \selectlanguage{french}%
11.80\selectlanguage{american}%
\tabularnewline
\hline 
\end{tabular}\foreignlanguage{english}{{\small{}\protect\caption{\label{tab:exp_res_scalar_ode1}Results obtained for the linear model}
}}
\par\end{centering}{\small \par}

\selectlanguage{english}%
\end{table}

\par\end{center}{\small \par}

The obtained results are presented in table \ref{tab:exp_res_scalar_ode1}.
In every cases, the Tracking estimator gives more precise estimation than NLS and GS estimators (both in term
of MSE and ARE), but this improvement is not impressive as MSE is
expressed at scale $10^{-5}$ and ARE at scale $10^{-3}$. The differences
are small among the different estimation methods and estimations are
reliable in all cases (even for the GS approach). This example mainly
illustrates tracking approach is a relevant estimation method and can
compete with the most used methods for simple model.

\subsubsection{Nonlinear w.r.t parameters}

\paragraph{Well-specified model}
We consider a following time dependent 
model
\begin{equation}
\dot{x}=\frac{\theta_{1}}{\theta_{2}^{2}+t}x.\label{eq:assumed_model}
\end{equation}
that is non-linear in parameters. We test two sample sizes $n=50$ and $n=20$
(observations are uniformly distributed between $0$ and $15$) and
two noise levels $\sigma=2$ and $\sigma=4$. The 
true parameter value is $\theta^{*}=\left(\theta_{1}^{*},\theta_{2}^{*}\right)=\left(1.4,1\right)$
and the initial condition is equal to $X_{0}^{*}=1$. The sequence of 
lambda used is $\lambda^{v}=\left\{ 10^{k}\right\} _{-1\leq k\leq7}$.

\begin{center}
{\small{}}
\begin{table}[h]
\selectlanguage{american}%
\centering{}{\small{}}%
\begin{tabular}{|c|c|c|c|c|}
\hline 
\selectlanguage{english}%
{\small{}$\left(n,\sigma\right)$}\selectlanguage{american}%
 & \selectlanguage{english}%
\selectlanguage{american}%
 & \multicolumn{1}{c|}{\selectlanguage{english}%
{\small{}MSE ($\times10^{-2}$)}\selectlanguage{american}%
} & \selectlanguage{english}%
{\small{}ARE ($\times10^{-2}$)}\selectlanguage{american}%
 & \selectlanguage{english}%
{\small{}Pred error }\selectlanguage{american}%
\tabularnewline
\hline 
\hline 
\selectlanguage{english}%
{\small{}$\left(50,2\right)$}\selectlanguage{american}%
 & \selectlanguage{english}%
\selectlanguage{american}%
 & \selectlanguage{english}%
\selectlanguage{american}%
 & \selectlanguage{english}%
\selectlanguage{american}%
 & \selectlanguage{english}%
\selectlanguage{american}%
\tabularnewline
\hline 
\selectlanguage{english}%
\selectlanguage{american}%
 & \selectlanguage{english}%
{\small{}$\widehat{\theta}^{T}$}\selectlanguage{american}%
 & \selectlanguage{english}%
{\small{}0.18}\selectlanguage{american}%
 & \selectlanguage{english}%
{\small{}4.04}\selectlanguage{american}%
 & \selectlanguage{french}%
7.79\selectlanguage{american}%
\tabularnewline
\hline 
\selectlanguage{english}%
\selectlanguage{american}%
 & \selectlanguage{english}%
{\small{}$\widehat{\theta}^{NLS}$}\selectlanguage{american}%
 & \selectlanguage{english}%
{\small{}0.18}\selectlanguage{american}%
 & \selectlanguage{english}%
{\small{}4.06}\selectlanguage{american}%
 & \selectlanguage{french}%
7.79\selectlanguage{american}%
\tabularnewline
\hline 
\selectlanguage{english}%
\selectlanguage{american}%
 & \selectlanguage{english}%
{\small{}$\widehat{\theta}^{GS}$}\selectlanguage{american}%
 & \selectlanguage{french}%
4.88\selectlanguage{american}%
 & \selectlanguage{french}%
19.94\selectlanguage{american}%
 & \selectlanguage{french}%
46.27\selectlanguage{american}%
\tabularnewline
\hline 
\selectlanguage{english}%
{\small{}$\left(50,4\right)$}\selectlanguage{american}%
 & \selectlanguage{english}%
\selectlanguage{american}%
 & \selectlanguage{english}%
\selectlanguage{american}%
 & \selectlanguage{english}%
\selectlanguage{american}%
 & \selectlanguage{english}%
\selectlanguage{american}%
\tabularnewline
\hline 
\selectlanguage{english}%
\selectlanguage{american}%
 & \selectlanguage{english}%
{\small{}$\widehat{\theta}^{T}$}\selectlanguage{american}%
 & \selectlanguage{english}%
{\small{}0.68}\selectlanguage{american}%
 & \selectlanguage{english}%
{\small{}8.35}\selectlanguage{american}%
 & \selectlanguage{french}%
15.57\selectlanguage{american}%
\tabularnewline
\hline 
\selectlanguage{english}%
\selectlanguage{american}%
 & \selectlanguage{english}%
{\small{}$\widehat{\theta}^{NLS}$}\selectlanguage{american}%
 & \selectlanguage{english}%
{\small{}0.68}\selectlanguage{american}%
 & \selectlanguage{english}%
{\small{}8.35}\selectlanguage{american}%
 & \selectlanguage{french}%
15.57\selectlanguage{american}%
\tabularnewline
\hline 
\selectlanguage{english}%
\selectlanguage{american}%
 & \selectlanguage{english}%
{\small{}$\widehat{\theta}^{GS}$}\selectlanguage{american}%
 & \selectlanguage{french}%
9.15\selectlanguage{american}%
 & \selectlanguage{french}%
30.10\selectlanguage{american}%
 & \selectlanguage{french}%
67.53\selectlanguage{american}%
\tabularnewline
\hline 
\selectlanguage{english}%
{\small{}$\left(20,2\right)$}\selectlanguage{american}%
 & \selectlanguage{english}%
\selectlanguage{american}%
 & \selectlanguage{english}%
\selectlanguage{american}%
 & \selectlanguage{english}%
\selectlanguage{american}%
 & \selectlanguage{english}%
\selectlanguage{american}%
\tabularnewline
\hline 
\selectlanguage{english}%
\selectlanguage{american}%
 & \selectlanguage{english}%
{\small{}$\widehat{\theta}^{T}$}\selectlanguage{american}%
 & \selectlanguage{english}%
{\small{}0.87}\selectlanguage{american}%
 & \selectlanguage{english}%
{\small{}8.93}\selectlanguage{american}%
 & \selectlanguage{french}%
15.59\selectlanguage{american}%
\tabularnewline
\hline 
\selectlanguage{english}%
\selectlanguage{american}%
 & \selectlanguage{english}%
{\small{}$\widehat{\theta}^{NLS}$}\selectlanguage{american}%
 & \selectlanguage{english}%
{\small{}0.87}\selectlanguage{american}%
 & \selectlanguage{english}%
{\small{}8.95}\selectlanguage{american}%
 & \selectlanguage{french}%
15.59\selectlanguage{american}%
\tabularnewline
\hline 
\selectlanguage{english}%
\selectlanguage{american}%
 & \selectlanguage{english}%
{\small{}$\widehat{\theta}^{GS}$}\selectlanguage{american}%
 & \selectlanguage{french}%
11.11\selectlanguage{american}%
 & \selectlanguage{french}%
33.43\selectlanguage{american}%
 & \selectlanguage{french}%
82.57\selectlanguage{american}%
\tabularnewline
\hline 
\selectlanguage{english}%
{\small{}$\left(20,4\right)$}\selectlanguage{american}%
 & \selectlanguage{english}%
\selectlanguage{american}%
 & \selectlanguage{english}%
\selectlanguage{american}%
 & \selectlanguage{english}%
\selectlanguage{american}%
 & \selectlanguage{english}%
\selectlanguage{american}%
\tabularnewline
\hline 
\selectlanguage{english}%
\selectlanguage{american}%
 & \selectlanguage{english}%
{\small{}$\widehat{\theta}^{T}$}\selectlanguage{american}%
 & \selectlanguage{english}%
{\small{}1.30}\selectlanguage{american}%
 & \selectlanguage{english}%
{\small{}11.16}\selectlanguage{american}%
 & \selectlanguage{french}%
15.64\selectlanguage{american}%
\tabularnewline
\hline 
\selectlanguage{english}%
\selectlanguage{american}%
 & \selectlanguage{english}%
{\small{}$\widehat{\theta}^{NLS}$}\selectlanguage{american}%
 & \selectlanguage{english}%
{\small{}1.33}\selectlanguage{american}%
 & \selectlanguage{english}%
{\small{}11.22}\selectlanguage{american}%
 & \selectlanguage{french}%
15.63\selectlanguage{american}%
\tabularnewline
\hline 
\selectlanguage{english}%
\selectlanguage{american}%
 & \selectlanguage{english}%
{\small{}$\widehat{\theta}^{GS}$}\selectlanguage{american}%
 & \selectlanguage{french}%
11.23\selectlanguage{american}%
 & \selectlanguage{french}%
32.79\selectlanguage{american}%
 & \selectlanguage{french}%
62.79\selectlanguage{american}%
\tabularnewline
\hline 
\end{tabular}\foreignlanguage{english}{{\small{}\protect\caption{\label{tab:exp_res_scalar_ode2}Results obtained for the non linear
model}
}}\selectlanguage{english}%
\end{table}

\par\end{center}{\small \par}
The results are presented in table \ref{tab:exp_res_scalar_ode2}.
The Tracking and NLS estimators have equivalent performance in the
well specified case. The GS approach gives a less precise estimation
but it is still reliable, but the prediction error for GS is much more important than for the Tracking and NLS estimators. This illustrates the high
sensitivity of ODE model w.r.t parameter and the need of accurate
estimates for prediction, even for simple models.

\paragraph{Misspecified model}

The data is generated by a pertubed model
\begin{equation}
\dot{x}=\frac{\theta_{1}}{\theta_{2}^{2}+t}x+\sin(t)\label{eq:true_model}
\end{equation}
with $\theta_{1}^{*}=1.4,\theta_{2}^{*}=1$ and $X_{0}^{*}=1$. Nevertheless, we still use the model (\ref{eq:assumed_model}) for the  parameter estimation. We use the two sample size $n=100$ and $n=50$ (observations were uniformly distributed
between $0$ and $15$) and the two  noise levels $\sigma=2$ and $\sigma=4$. 
We use a sequence of hyperparameter $\lambda^{v}=\left\{ 10^{k},\,5\times10^{k}\right\} _{-4\leq k\leq0}$. 

Moreover, we are interested in the use of the residual control
$\overline{u}$ obtained along the parametric estimation in order
to propose a "corrected" model: 
\begin{equation}
\dot{x}=\frac{\theta_{1}}{\theta_{2}^{2}+t}x+\overline{u}\label{eq:corrected_model}
\end{equation}
for minimizing the ``corrected'' prediction error 
\begin{equation}
\mathrm{\mathbb{E}_{\theta^{\text{*}},\sigma}\left[\left\Vert Y-X_{\widehat{\theta},\overline{u}}\right\Vert _{L^{2}}^{2}\right]}\label{eq:corr_prediction_error}
\end{equation}
When model misspecification is suspected, we can consider two trajectory predictors $X_{\widehat{\theta},0}$ and  $X_{\widehat{\theta},\overline{u}}$ for a given $\lambda$. From the definition of the criterion $S$, the corrected trajectory $X_{\widehat{\theta},\overline{u}}$ is prone to give smaller prediction errors than the misspecified trajectory $X_{\widehat{\theta},0}$. This indicates that we should select the hyperparameter $\lambda$ in a different way in the case of misspecification, if we want to use the correction $\overline{u}$ that depends also on $\lambda$. For this reason, we propose to select $\lambda$ by minimizing the Corrected Sum of Squared Errors 
\[CSSE(\lambda) = \sum_{i=1}^n\left(Y_{i}-X_{\widehat{\theta_{\lambda}^{T}},\overline{u}}(t_{i})\right)^{2}.\]
as a proxy for the prediction errror (\ref{eq:corr_prediction_error}). 
Finally, we have two tracking estimators $\widehat{\theta}^T$ and $\widehat{\theta}_c^T$  based on two choices of hyperparameter ($\lambda$ that minimizes $SSE$ and $\lambda_c$ that minimizes $CSSE$). 
The estimation of the perturbation (or control $u$) is done in two ways: 
\begin{itemize}
\item for the Tracking estimator  $\widehat{\theta}^{T}$, it is provided directly by the estimation process;
\item for the NLS estimator $\widehat{\theta}^{NLS}$, we propose to estimate the perturbation based on the nonparametric proxy $\widehat{X}$ 
\[
\overline{u}(t)=\dot{\widehat{X}}(t)-\frac{\widehat{\theta_{1}}}{\widehat{\theta_{2}}^{2}+t}\widehat{X}(t).
\]
\end{itemize}
For the Generalized Smoothing, we do not have to estimate the perturbation $u$, as the penalized spline $\widehat{X}(\cdot,\hat{\theta}^{GS})$ already contains the model misspecification. 

\begin{center}
{\small{}}
\begin{table}[h]
\centering{}{\small{}}%
\begin{tabular}{|c|c|c|c|c|c|}
\hline 
{\small{}$(n,\sigma)$} &  & {\small{}MSE ($\times10^{-2}$)} & {\small{}ARE ($\times10^{-2}$)} & {\small{}Pred error } & {\small{} Corrected Pred error }\tabularnewline
\hline 
\hline 
{\small{}$(100,2)$} &  &  &  &  & \tabularnewline
\hline 
 & $\widehat{\theta}^{T}$ & {\small{}4.31} & {\small{}24.16} & {\small{}8.51} & {\small{}8.16}\tabularnewline
\hline 
 & $\widehat{\theta^{c}}^{T}$ & \selectlanguage{french}%
2.88\selectlanguage{english}%
 & \selectlanguage{french}%
18.52\selectlanguage{english}%
 & \selectlanguage{french}%
16.94\selectlanguage{english}%
 & \selectlanguage{french}%
8.03\selectlanguage{english}%
\tabularnewline
\hline 
 & {\small{}$\widehat{\theta}^{NLS}$} & {\small{}4.61} & {\small{}25.11} & {\small{}8.15} & {\small{}8.09}\tabularnewline
\hline 
 & {\small{}$\widehat{\theta}^{GS}$} & \selectlanguage{french}%
2.39\selectlanguage{english}%
 & \selectlanguage{french}%
14.27\selectlanguage{english}%
 & \selectlanguage{french}%
42.36\selectlanguage{english}%
 & \selectlanguage{french}%
42.27\selectlanguage{english}%
\tabularnewline
\hline 
{\small{}$(50,2)$} &  &  &  &  & \tabularnewline
\hline 
 & $\widehat{\theta}^{T}$ & {\small{}4.35} & {\small{}24.05} & {\small{}8.84} & {\small{}8.20}\tabularnewline
\hline 
 & $\widehat{\theta^{c}}^{T}$ & \selectlanguage{french}%
3.31\selectlanguage{english}%
 & \selectlanguage{french}%
19.18\selectlanguage{english}%
 & \selectlanguage{french}%
20.83\selectlanguage{english}%
 & \selectlanguage{french}%
8.21\selectlanguage{english}%
\tabularnewline
\hline 
 & {\small{}$\widehat{\theta}^{NLS}$} & {\small{}4.62} & {\small{}24.89} & {\small{}8.20} & {\small{}8.27}\tabularnewline
\hline 
 & {\small{}$\widehat{\theta}^{GS}$} & \selectlanguage{french}%
4.64\selectlanguage{english}%
 & \selectlanguage{french}%
19.87\selectlanguage{english}%
 & \selectlanguage{french}%
55.75\selectlanguage{english}%
 & \selectlanguage{french}%
55.63\selectlanguage{english}%
\tabularnewline
\hline 
{\small{}$(100,4)$} &  &  &  &  & \tabularnewline
\hline 
 & $\widehat{\theta}^{T}$ & {\small{}4.74} & {\small{}24.95} & {\small{}16.14} & {\small{}15.75}\tabularnewline
\hline 
 & $\widehat{\theta^{c}}^{T}$ & \selectlanguage{french}%
4.36\selectlanguage{english}%
 & \selectlanguage{french}%
21.75\selectlanguage{english}%
 & \selectlanguage{french}%
25.23\selectlanguage{english}%
 & \selectlanguage{french}%
15.83\selectlanguage{english}%
\tabularnewline
\hline 
 & {\small{}$\widehat{\theta}^{NLS}$} & {\small{}4.99} & {\small{}25.67} & {\small{}15.76} & {\small{}15.82}\tabularnewline
\hline 
 & {\small{}$\widehat{\theta}^{GS}$} & \selectlanguage{french}%
8.08\selectlanguage{english}%
 & \selectlanguage{french}%
27.04\selectlanguage{english}%
 & \selectlanguage{french}%
70.25\selectlanguage{english}%
 & \selectlanguage{french}%
70.09\selectlanguage{english}%
\tabularnewline
\hline 
{\small{}$(50,4)$} &  &  &  &  & \tabularnewline
\hline 
 & $\widehat{\theta}^{T}$ & {\small{}4.83} & {\small{}24.77} & {\small{}16.51} & {\small{}15.84}\tabularnewline
\hline 
 & $\widehat{\theta^{c}}^{T}$ & \selectlanguage{french}%
6.10\selectlanguage{english}%
 & \selectlanguage{french}%
25.67\selectlanguage{english}%
 & \selectlanguage{french}%
27.73\selectlanguage{english}%
 & \selectlanguage{french}%
16.01\selectlanguage{english}%
\tabularnewline
\hline 
 & {\small{}$\widehat{\theta}^{NLS}$} & {\small{}5.05} & {\small{}25.49} & {\small{}15.83} & {\small{}16.03}\tabularnewline
\hline 
 & {\small{}$\widehat{\theta}^{GS}$} & \selectlanguage{french}%
9.18\selectlanguage{english}%
 & \selectlanguage{french}%
29.54\selectlanguage{english}%
 & \selectlanguage{french}%
82.82\selectlanguage{english}%
 & \selectlanguage{french}%
82.72\selectlanguage{english}%
\tabularnewline
\hline 
\end{tabular}{\small{}\protect\caption{\label{tab:misspecified_model_results}Estimation results for misspecified
model}
}
\end{table}

\par\end{center}{\small \par}

Results are presented in table \ref{tab:misspecified_model_results}.
The two first columns gives the parametric estimation performance
in terms of MSE and ARE. The third column gives an estimation of (\ref{eq:prediction_error})
and the fourth one an estimation of (\ref{eq:corr_prediction_error}). 

We can see $\widehat{\theta}^{T}$ gives more accurate parametric
estimation than the NLS estimator.  The use of residual control in  $X_{\widehat{\theta},\overline{u}}$ improves the prediction error in
any case. At the contrary, the correction of the NLS estimate with the estimated control 
makes things worst, which can be explained by the use of the non-parametric
estimator of the derivative. The Generalized Smoothing estimator $\widehat{\theta}^{GS}$ competes well
with others approaches, that can be explained by the relaxation introduced
by the collocation which makes the method robust in misspecification
presence. However, we observe a fast drop in estimation precision w.r.t
noise augmentation and poor performance for prediction purpose.

The corrected estimator $\widehat{\theta^{c}}^{T}$ gives higher
precision than $\widehat{\theta}^{T}$ for small measurement error $\sigma$. 
We also notice the dramatic drop in prediction error by using
the corrected model instead of the initial one. It is due to the fact
we effectively take into account an exogeneous perturbation using
$\overline{u}$ for parametric estimation. We simultaneously estimate
the parametric part and the nonparametric part of the model as it is 
the case. But we can not propose an adaptive way to detect when the nonparametric model correction
 $\overline{u}$ should be applied, as it is far beyond the scope of that paper.

Finally, we are also interested in the control $\overline{u}$ itself,
even though we do not expect it gives us an estimation of the true
control $u^{*}(t)=\sin(t)$ (because of identifiability issues). Its
features can give hints about potential misspecification presence
and qualitative informations about its nature. We plot in figure \ref{fig:mean_residual_control}
the mean control obtained in the case $\left(n,\sigma\right)=\left(100,\,2\right)$
in blue and the true control $u^{*}(t)=\sin(t)$ in green for the
sake of comparison. Although, the scale is not the same, we can see that the estimated 
control exhibits some features of the $u^{*}$, such as oscillations, with the same approximate period. The pertubation could be  
used as an exploratory tool for analyzing and inferring the missing part in the dynamics.

\begin{center}
\begin{figure}[h]
\centering{}\includegraphics[scale=0.4]{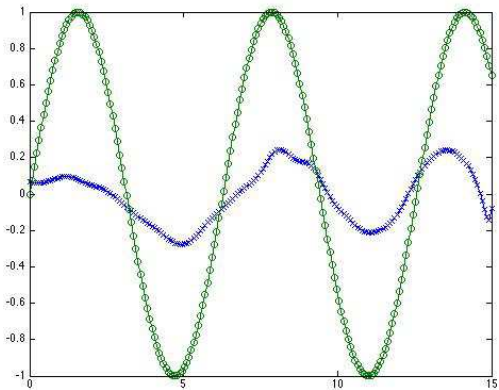}\protect\caption{\label{fig:mean_residual_control}Obtained mean residual control for
$\left(n,\sigma\right)=\left(100,\,2\right)$}
\end{figure}
\par\end{center}

\subsection{$\alpha-$Pinene model}

The "$\alpha-$Pinene model" is a model introduced in \cite{Rodriguez2006} for modeling the isomerization of $\alpha-$Pinene. 
It is an autonomous linear ODE in $\mathbb{R}^5$ with a sparse structure
\[
\dot{X}(t)=\left(\begin{array}{ccccc}
-(\theta_{1}+\theta_{2}) & 0 & 0 & 0 & 0\\
\theta_{1} & 0 & 0 & 0 & 0\\
\theta_{2} & 0 & -(\theta_{3}+\theta_{4}) & 0 & \theta_{5}\\
0 & 0 & \theta_{3} & 0 & 0\\
0 & 0 & \theta_{4} & 0 & -\theta_{5}
\end{array}\right)X(t):=A(\theta)X(t)
\]
It is an Initial Value Problem, with known initial condition $X_{0}^{*}=(100,0,0,0,0)$ . This
estimation problem is still considered as cumbersome and many estimation
methods fails to converge or converge to bad local solutions because
of high correlation between $\theta_{4}$ and $\theta_{5}$.
Before analyzing the real dataset, we perform a simulation study for evaluating the difficulty of the estimation problem, and benchmarking the several estimators. 

\subsubsection{Simulated data}
The observation interval is  $\left[0,\,100\right]$ and the true parameter is $\theta^{*}=(5.93,\,$ $2.96,\,2.05,$ \,27.5,$\,4)\times10^{-4}$.
Because of dramatic differences in the order of magnitude of the state variables, 
the noise standard deviation has to be rescaled componentwise. Here
for a given $\sigma$ value the standard deviation applied to the
state variable $X_{i}$ is equal to $\frac{\sigma}{100}\times\frac{1}{T}\int_{0}^{T}X_{i}(t)dt$,
$\,\frac{1}{T}\int_{0}^{T}X_{i}(t) dt$ being the $X_{i}$ mean value
on the observation interval. For the Tracking estimator, we use a sequence of hyperparameter
$\lambda^{v}=\left\{ 10^{k},\,5\times10^{k}\right\} _{0\leq k\leq5}$. 

\begin{center}
{\small{}}
\begin{table}[h]
\selectlanguage{american}%
\centering{}{\small{}}%
\begin{tabular}{|c|c|c|c|c|}
\hline 
\selectlanguage{english}%
{\small{}$\left(n,\sigma\right)$}\selectlanguage{american}%
 & \selectlanguage{english}%
\selectlanguage{american}%
 & \multicolumn{1}{c|}{\selectlanguage{english}%
{\small{}MSE ($\times10^{-4}$)}\selectlanguage{american}%
} & \selectlanguage{english}%
{\small{}ARE ($\times10^{-2}$)}\selectlanguage{american}%
 & \selectlanguage{english}%
{\small{}Pred error}\selectlanguage{american}%
\tabularnewline
\hline 
\hline 
\selectlanguage{english}%
{\small{}$\left(100,8\right)$}\selectlanguage{american}%
 & \selectlanguage{english}%
\selectlanguage{american}%
 & \selectlanguage{english}%
\selectlanguage{american}%
 & \selectlanguage{english}%
\selectlanguage{american}%
 & \selectlanguage{english}%
\selectlanguage{american}%
\tabularnewline
\hline 
\selectlanguage{english}%
\selectlanguage{american}%
 & \selectlanguage{english}%
{\small{}$\widehat{\theta}^{T}$}\selectlanguage{american}%
 & \selectlanguage{english}%
{\small{}2.83}\selectlanguage{american}%
 & \selectlanguage{english}%
{\small{}1.25}\selectlanguage{american}%
 & \selectlanguage{french}%
52.72\selectlanguage{american}%
\tabularnewline
\hline 
\selectlanguage{english}%
\selectlanguage{american}%
 & \selectlanguage{english}%
{\small{}$\widehat{\theta}^{NLS}$}\selectlanguage{american}%
 & \selectlanguage{english}%
{\small{}2.75}\selectlanguage{american}%
 & \selectlanguage{english}%
{\small{}0.87}\selectlanguage{american}%
 & \selectlanguage{french}%
51.71\selectlanguage{american}%
\tabularnewline
\hline 
\selectlanguage{english}%
\selectlanguage{american}%
 & \selectlanguage{english}%
{\small{}$\widehat{\theta}^{GS}$}\selectlanguage{american}%
 & \selectlanguage{french}%
3.21\selectlanguage{american}%
 & \selectlanguage{french}%
1.72\selectlanguage{american}%
 & \selectlanguage{french}%
54.18\selectlanguage{american}%
\tabularnewline
\hline 
\selectlanguage{english}%
{\small{}$\left(50,8\right)$}\selectlanguage{american}%
 & \selectlanguage{english}%
\selectlanguage{american}%
 & \selectlanguage{english}%
\selectlanguage{american}%
 & \selectlanguage{english}%
\selectlanguage{american}%
 & \selectlanguage{english}%
\selectlanguage{american}%
\tabularnewline
\hline 
\selectlanguage{english}%
\selectlanguage{american}%
 & \selectlanguage{english}%
{\small{}$\widehat{\theta}^{T}$}\selectlanguage{american}%
 & \selectlanguage{english}%
{\small{}2.92}\selectlanguage{american}%
 & \selectlanguage{english}%
{\small{}1.35}\selectlanguage{american}%
 & \selectlanguage{french}%
52.40\selectlanguage{american}%
\tabularnewline
\hline 
\selectlanguage{english}%
\selectlanguage{american}%
 & \selectlanguage{english}%
{\small{}$\widehat{\theta}^{NLS}$}\selectlanguage{american}%
 & \selectlanguage{english}%
{\small{}3.14}\selectlanguage{american}%
 & \selectlanguage{english}%
{\small{}1.24}\selectlanguage{american}%
 & \selectlanguage{french}%
52.20\selectlanguage{american}%
\tabularnewline
\hline 
\selectlanguage{english}%
\selectlanguage{american}%
 & \selectlanguage{english}%
{\small{}$\widehat{\theta}^{GS}$}\selectlanguage{american}%
 & \selectlanguage{french}%
7.05\selectlanguage{american}%
 & \selectlanguage{french}%
2.97\selectlanguage{american}%
 & \selectlanguage{french}%
54.54\selectlanguage{american}%
\tabularnewline
\hline 
\selectlanguage{english}%
{\small{}$\left(100,16\right)$}\selectlanguage{american}%
 & \selectlanguage{english}%
\selectlanguage{american}%
 & \selectlanguage{english}%
\selectlanguage{american}%
 & \selectlanguage{english}%
\selectlanguage{american}%
 & \selectlanguage{english}%
\selectlanguage{american}%
\tabularnewline
\hline 
\selectlanguage{english}%
\selectlanguage{american}%
 & \selectlanguage{english}%
{\small{}$\widehat{\theta}^{T}$}\selectlanguage{american}%
 & \selectlanguage{english}%
{\small{}6.08}\selectlanguage{american}%
 & \selectlanguage{english}%
{\small{}3.58}\selectlanguage{american}%
 & \selectlanguage{french}%
103.73\selectlanguage{american}%
\tabularnewline
\hline 
\selectlanguage{english}%
\selectlanguage{american}%
 & \selectlanguage{english}%
{\small{}$\widehat{\theta}^{NLS}$}\selectlanguage{american}%
 & \selectlanguage{english}%
{\small{}14}\selectlanguage{american}%
 & \selectlanguage{english}%
{\small{}4.8}\selectlanguage{american}%
 & \selectlanguage{french}%
103.38\selectlanguage{american}%
\tabularnewline
\hline 
\selectlanguage{english}%
\selectlanguage{american}%
 & \selectlanguage{english}%
{\small{}$\widehat{\theta}^{GS}$}\selectlanguage{american}%
 & \selectlanguage{french}%
15.4\selectlanguage{american}%
 & \selectlanguage{french}%
4.06\selectlanguage{american}%
 & \selectlanguage{french}%
106.88\selectlanguage{american}%
\tabularnewline
\hline 
\selectlanguage{english}%
{\small{}$\left(50,16\right)$}\selectlanguage{american}%
 & \selectlanguage{english}%
\selectlanguage{american}%
 & \selectlanguage{english}%
\selectlanguage{american}%
 & \selectlanguage{english}%
\selectlanguage{american}%
 & \selectlanguage{english}%
\selectlanguage{american}%
\tabularnewline
\hline 
\selectlanguage{english}%
\selectlanguage{american}%
 & \selectlanguage{english}%
{\small{}$\widehat{\theta}^{T}$}\selectlanguage{american}%
 & \selectlanguage{english}%
{\small{}11}\selectlanguage{american}%
 & \selectlanguage{english}%
{\small{}7.01}\selectlanguage{american}%
 & \selectlanguage{french}%
103.91\selectlanguage{american}%
\tabularnewline
\hline 
\selectlanguage{english}%
\selectlanguage{american}%
 & \selectlanguage{english}%
{\small{}$\widehat{\theta}^{NLS}$}\selectlanguage{american}%
 & \selectlanguage{english}%
{\small{}26}\selectlanguage{american}%
 & \selectlanguage{english}%
{\small{}8.47}\selectlanguage{american}%
 & \selectlanguage{french}%
103.58\selectlanguage{american}%
\tabularnewline
\hline 
\selectlanguage{english}%
\selectlanguage{american}%
 & \selectlanguage{english}%
{\small{}$\widehat{\theta}^{GS}$}\selectlanguage{american}%
 & \selectlanguage{french}%
25.85\selectlanguage{american}%
 & \selectlanguage{french}%
5.36\selectlanguage{american}%
 & \selectlanguage{french}%
107.91\selectlanguage{american}%
\tabularnewline
\hline 
\end{tabular}\foreignlanguage{english}{{\small{}\protect\caption{\label{tab:exp_res_a_pinene_3param}Results obtained for $\alpha-$Pinene}
}}\selectlanguage{english}%
\end{table}

\par\end{center}{\small \par}

The results are presented in table \ref{tab:exp_res_a_pinene_3param}.
The Tracking and Generalized Smoothing estimators give more accurate parameter
estimation than the Nonlinear Least Squares. In the last case, GS gives the best
performance in terms of ARE. In this model, the Relative Error is especially
relevant to quantify parameter precision because of important differences
between the scale of $\theta_{4}$ and the other parameters. As expected the difference
in performance mainly comes from the estimation of the couple $\left(\theta_{4},\theta_{5}\right)$.
This model shows that the NLS is appropriate for parameter
estimation whereas the GS seems to favor the quality of prediction, showing that estimation and prediction are somehow two competing objectives.
The Tracking estimator realizes a trade-off between these two objectives. 

\subsubsection{Real data analysis}

We use the data coming from \cite{Fuguitt1947}, and  presented in table \ref{tab:exp_data}.
They consist in simultaneous measures of the 5 components relative concentration
at eight time steps. We compare our results with the previous estimation  $\widehat{\theta_{b}}=10^{-4}\times\left(0.593,\,0.296,\,0.205,\,2.75,\,0.4\right)$ obtained in \cite{Rodriguez2006}, which provides a good data fitting. We use the Tracking method for the estimation of the parameter and of
 the residual control $\overline{u}$.  In order to avoid numerical problems, we have divided the observation by $1000$ and renormalized the parameters (as the system is autonomous). 
%Because of the previous estimation we expect to find low order value
%for the parameter estimation this can drive to numerical issues for
%the optimization algorithm (i.e stopping too early because of weak
%change in parameter value update). To circumvent that and because
%the model is autonomous we have divided the observation time by $1000$. 

For the non-parametric estimator $\widehat{X}$, we use only
two nodes: one at the end and one at the beginning of the observation
interval. We have also impose to the non-parametric estimator to start
from the known initial condition $X_{0}^{*}=(100,\,0,\,0,\,0,\,0)$.
For the Tracking estimator, we use  $\lambda=10^{k},\, k=1,2\cdots,11$
and $\lambda=5\times10^{k},\, k=1,2,3,4$ and we select the hyper-parameter that minimizes
the SSE (named $\lambda_{SSE}$. We call the corresponding estimator
$\widehat{\theta_{SSE}}^{T}$) . The estimation results are presented
in table \ref{tab:real_case_parametric_est}.

\begin{center}
{\small{}}
\begin{table}[h]
\centering{}{\small{}}%
\begin{tabular}{|c|c|c|c|c|c|}
%\hline  & \multicolumn{5}{c|}{}\tabularnewline
\hline 
%\hline 
{\small{}Times (min)} & {\small{}$X_{1}$} & {\small{}$X_{2}$} & {\small{}$X_{3}$} & {\small{}$X_{4}$} & {\small{}$X_{5}$}\tabularnewline
\hline 
{\small{}1230} & {\small{}88.35} & {\small{}7.3} & {\small{}2.3} & {\small{}0.4} & {\small{}1.75}\tabularnewline
\hline 
{\small{}3060} & {\small{}76.4} & {\small{}15.6} & {\small{}4.5} & {\small{}0.7} & {\small{}2.8}\tabularnewline
\hline 
{\small{}4920} & {\small{}65.1} & {\small{}23.1} & {\small{}5.3} & {\small{}1.1} & {\small{}5.8}\tabularnewline
\hline 
{\small{}7800} & {\small{}50.4} & {\small{}32.9} & {\small{}6} & {\small{}1.5} & {\small{}9.3}\tabularnewline
\hline 
{\small{}10680} & {\small{}37.5} & {\small{}42.7} & {\small{}6} & {\small{}1.9} & {\small{}12}\tabularnewline
\hline 
{\small{}15030} & {\small{}25.9} & {\small{}49.1} & {\small{}5.9} & {\small{}2.2} & {\small{}17}\tabularnewline
\hline 
{\small{}22620} & {\small{}14} & {\small{}57.4} & {\small{}5.1} & {\small{}2.6} & {\small{}21}\tabularnewline
\hline 
{\small{}36420} & {\small{}4.5} & {\small{}63.1} & {\small{}3.8} & {\small{}2.9} & {\small{}25.7}\tabularnewline
\hline 
\end{tabular}{\small{}\protect\caption{\label{tab:exp_data}Experimental data for $\alpha-$pinene model
coming from Fuguitt \&Hawkins }
}
\end{table}
\par\end{center}{\small \par}

\begin{center}
{\small{}}
\begin{table}[h]
\centering{}{\small{}}%
\begin{tabular}{|c|c|c|c|c|c|c|}
%\hline & \multicolumn{6}{c|}{}\tabularnewline
\hline 
%\hline 
{\small{}$10^{-4}$} & {\small{}$\theta_{1}$} & {\small{}$\theta_{2}$} & {\small{}$\theta_{3}$} & {\small{}$\theta_{4}$} & {\small{}$\theta_{5}$} & {\small{}SSE}\tabularnewline
\hline 
{\small{}$\widehat{\theta_{SSE}}^{T}$} & {\small{}0.589} & {\small{}0.290} & {\small{}0.193} & {\small{}2.301} & {\small{}0.234} & {\small{}23.88}\tabularnewline
\hline 
{\small{}$\widehat{\theta}_{b}$} & {\small{}0.593} & {\small{}0.296} & {\small{}0.205} & {\small{}2.75} & {\small{}0.4} & {\small{}19.89}\tabularnewline
\hline 
\end{tabular}{\small{}\protect\caption{\label{tab:real_case_parametric_est} Parameter estimates }
}
\end{table}
\par\end{center}{\small \par}

We have obtained $\lambda_{SSE}=100$. For $\lambda\geq 5000$, the estimated values are almost constant and equal to
$\theta=\left(0.583,\,0.295,\,0.207,\,2.259,\,0.238\right)$. The first three estimated parameters are close to the estimates given in \cite{Rodriguez2006}, but $\theta_{4}$ and $\theta_{5}$ are different; however, we obtain good predicted curves, similar to Rodriguez et al., see figure \ref{fig:plot_est_curve}.

\begin{center}
\begin{figure}[h]
\begin{centering}
\includegraphics[scale=0.5]{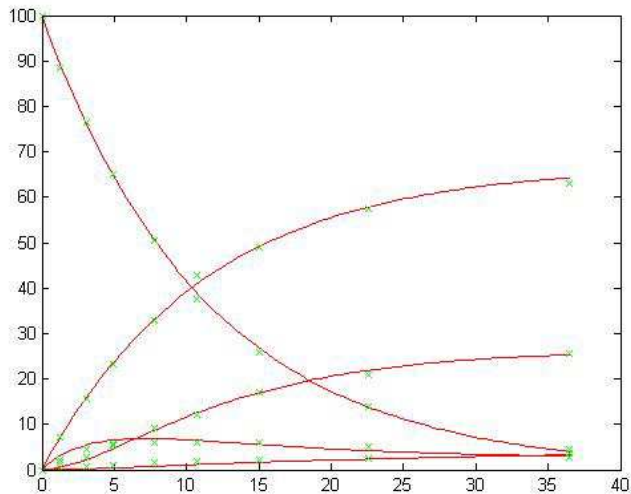}\protect\caption{\label{fig:plot_est_curve}Reconstructed curve for $\widehat{\theta_{SSE}^{T}}$
with data in green}
\par\end{centering}
\end{figure}
\par\end{center}

We can compute the minimal control $\overline{u}_{SSE}$ corresponding to $(\widehat{\theta_{SSE}}^{T},\,\lambda_{SSE})$, and  we can also compute the minimal control corresponding to $\theta=\widehat{\theta_{b}}$ for a given value of $\lambda$ (that is $\overline{u}_{\widehat{\theta_{b}},\lambda}$)
and compare its norm with $\overline{u}_{SSE}$ when $\lambda=\lambda_{SSE}$.  The controls are represented in figure \ref{fig:plot_estimated_control}
where the curves plotted with $\times$ represent the minimal control obtained
for $\widehat{\theta_{b}}$ and the curves plotted with $\circ$ are the control
obtained with $\widehat{\theta_{SSE}}^{T}$ . Here, the Tracking control is a
five-dimensional vector, where each entry $\overline{u}_{i}$ corresponds
to one state variable $X_{i}$. The control plotted in yellow corresponds
to $X_{1}$, the one in black correspond to $X_{2}$, the one in green
correspond to $X_{3}$, the one in blue correspond to $X_{4}$ and
the one in red to $X_{5}$.

\begin{center}
\begin{figure}[h]
\centering{}\includegraphics[scale=0.5]{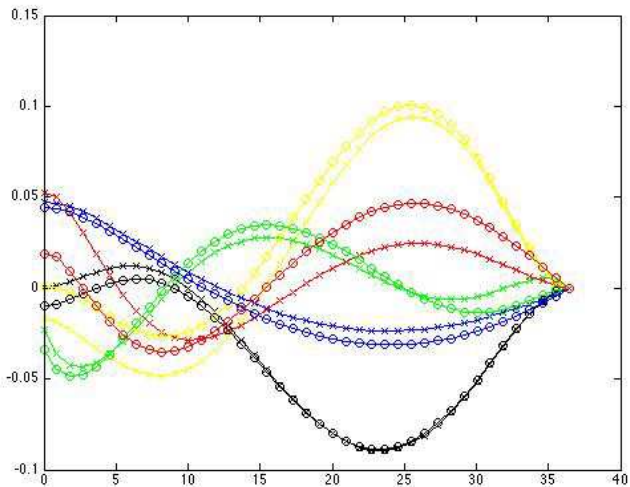}\protect\caption{\label{fig:plot_estimated_control}Estimated control for $\widehat{\theta_{b}}$
with $\lambda=\lambda_{SSE}$ and $\widehat{\theta_{SSE}}^{T}$ }
\end{figure}

\par\end{center}

As expected, the estimated control  $\overline{u}_{SSE}$ 
is smaller in $L^{2}-$ norm for $\widehat{\theta_{SSE}}^{T}$ than
for $\widehat{\theta_{b}}$. Nonetheless according to \ref{fig:plot_estimated_control}, there
is no clear difference between $\overline{u}_{\widehat{\theta_{b}},\lambda}$
and $\overline{u}_{SSE}$ except for their component related to $X_{5}$
(in red on the figure \ref{fig:plot_estimated_control}) and which
is the state variable exclusively related to the parameters $\theta_{4}$
and $\theta_{5}$ (the most difficult parameters to estimate according
to \cite{Rodriguez2006}. Even though the two resulting solutions  $X_{\widehat{\theta_{b}}}$ and $X_{\widehat{\theta_{SSE}}^{T}}$
are close, the insight given by $\overline{u}_{\widehat{\theta_{b}},\lambda}$
and $\overline{u}_{SSE}$ shows stronger differences at the dynamic
scale.

\section{Conclusion}

We have introduced a new estimation method for parameter estimation in linear ODE based on relaxation of the initial model. 
Similarly to the Generalized Smoothing estimator, we end up with a function $X_{\theta,u}$ that is an approximate solution of the ODE model of interest. The added perturbation $u$ enables to take into account the noisy observations but also some uncertainty in the model. Quite remarkably, the trade-off between model and data discrepancy is formulated and solved by using Optimal Control Theory and this work is one of the first use of this theory for dealing with statistical estimation and optimization in infinite dimensional spaces. Moreover, this functional framework and the Riccati theory gives practical algorithms and theoretical insight in the properties of the estimator that permits a detailed analysis of the statistical properties. 
In addition to the parameter estimate, we obtain also directly an estimate of the model discrepancy through the perturbation $u$ that can give hints for analyzing and discussing the relevancy of a parametric model. More can be done with the estimated  $u$ about model analysis, and complementary analysis about model testing could be done following the results given in \cite{HookerEllner2013}. 

In the experiments, we show that the Tracking estimator have similar or better performances than nonlinear least squares or generalized smoothing, even in the case of very simple models with closed-form expression. Hence, paradoxically, the use of perturbed model and of nonparametric estimators can ameliorate the statistical efficiency of standard estimates, even in well-specified cases. In the case of model misspecification, the differences are bigger, as the relaxation brought by $\lambda$ gives us a more robust estimation
method (comparing to NLS) which can deal with small model definition
imperfection. Moreover, the optimal control obtained for a given parameter
estimate allows us to minimize the prediction error by introducing
a proper correction term to the initial model. This control can also
be used as a qualitative tool to diagnose model misspecification.

However, we are aware of some  limitations of our method: first, we assume that the initial condition is known. We can consider $X_0^*$ as an additional parameter to estimate, and reformulate our approach for doing simultaneous observations. The second but the main limitation is the linear assumption about the ODE. Although linear ODEs are common in applications, numerous useful models are nonlinear and thus our methodology cannot be applied directly. Nevetheless, our work can be extended by using more general results of optimal control, such as Pontryagin Maximum Principle that can offer efficient characterization in the general nonlinear case. 
%Despite these issues, the performance of our estimator on the tested model, the benefits brought comparing to other methods drives us to
%look for such kind of improvements.

\newpage{}

\appendix

\part*{Appendix}

\section{Fundamentals Results of Optimal Control: Linear-Quadratic Theory }

The  "theorem and definition" in section 2.2 is a particular case of a more general theorem
which ensures existence and uniqueness of optimal control for cost under
the form: 
\begin{equation}
C\left(t_{0},u,\lambda\right)=z_{u}(T)^{T}Qz_{u}(T)+\int_{t_{0}}^{T}z_{u}(t)^{T}W(t)z_{u}(t)+u(t)^{T}U(t)u(t)dt \label{eq:GeneralCost}
\end{equation}

\begin{thm}
\label{thm:general_lq_theorem_existence_unicity}Let $A\in L^{2}(\left[0,\, T\right],\mathbb{R}^{d\times d})$
and $B\in L^{2}(\left[0,\, T\right],\mathbb{R}^{d\times d})$ We consider
$z_{u}$ the solution of the following ODE:
\[
\dot{z_{u}}(t)=A(t)z_{u}(t)+B(t)u(t),\: z(t_{0})=z_{0}
\]
and we want to minimize the cost (\ref{eq:GeneralCost})
%\[
%C\left(t_{0},u,\lambda\right)=z_{u}(T)^{T}Qz_{u}(T)+\int_{t_{0}}^{T}z_{u}(t)^{T}W(t)z_{u}(t)+u(t)^{T}U(t)u(t)dt
%\]
defined on $L^{2}(\left[0,\, T\right],\mathbb{R}^{d})$,  with $Q$ positive, $W\in L^{\infty}(\left[0,\, T\right],\mathbb{R}^{d\times d})$
positive matrix for all $t\in\left[0,\, T\right]$ and $U(t)$ definite
positive matrix for all $t\in\left[0,\, T\right]$ respecting the
coercivity condition: 
\[
\exists\alpha>0\, s.t\,\forall u\in L^{2}(\left[0,\, T\right],\mathbb{R}^{d})\,:\,\int_{0}^{T}u(t)^{T}U(t)u(t)dt\geq\alpha\int_{0}^{T}\left\Vert u(t)\right\Vert _{2}^{2}dt
\]
It exists a unique optimal trajectory $z_{\bar{u}}$ associated to
the unique optimal control $\overline{u}(t)=U^{-1}(t)E(t)B(t)z_{\overline{u}}(t)$
where $E$ is the matrix solution of the Riccati ODE:\textup{
\[
\begin{array}{l}
\dot{E}(t)=W(t)-A(t)^{t}E(t)-E(t)A(t)-E(t)B(t)U(t)^{-1}B(t)^{T}E(t)\\
E(T)=-Q
\end{array}
\]
} and the minimal cost is equal to: $C\left(t_{0},\overline{u},\lambda\right)=-z_{0}^{T}E(t_{0})z_{0}$. 

\newpage{}
\end{thm}

\section{Proof \& Intermediary results}

\subsection{$\theta\protect\longmapsto S(\widehat{X};\theta,\lambda)$ and $\theta\protect\longmapsto S(X^{*};\theta,\lambda)$
properties}
\begin{lem}
\label{lem:Ricatti_existence}Let us define $E$ the solution of 
\begin{equation}
\begin{array}{l}
\dot{E}(t)=W(t)-A(t)^{t}E(t)-E(t)A(t)-\frac{1}{\lambda}E(t)^{2}\\
E(T)=-Q
\end{array}\label{eq:classical_riccati_ODE_formulation_1}
\end{equation}
 with $A(t)\in L^{2}(\left[0,\, T\right],\mathbb{R}^{d\times d})$
$Q$ bounded,$W\in L^{\infty}(\left[0,\, T\right],\mathbb{R}^{d\times d})$. 

Then $E$ is bounded on $\left[0\,,\, T\right]$. \end{lem}
\begin{proof}
(This proof is presented in Sontag's book ``Mathematical Control
Theory'' \cite{Sontag1998} chapter 7 theorem 30) 

By using theorem \ref{thm:general_lq_theorem_existence_unicity} and
if we define the quadratic cost: 
\[
C(t_{0},u,\lambda)=x_{u}(T)^{T}Qx_{u}(T)+\int_{t_{0}}^{T}x_{u}(t)^{T}W(t)x_{u}(t)+\lambda\left\Vert u(t)\right\Vert _{2}^{2}dt
\]
with $x_{u}$ the ODE solution of 
\[
\begin{array}{l}
\dot{x_{u}}(t)=A(t)x_{u}(t)+u(t)\\
x_{u}(t_{0})=x_{0}
\end{array}
\]
We know that
\[
\min_{u}C(t_{0},u,\lambda)=-x_{0}^{T}E(t_{0})x_{0}
\]
Let us reason by contradiction: if $E(t)$is not bounded then $\exists t_{e}\in\left[0\,,\, T\right]$
s.t $\lim{}_{t\rightarrow t_{e}+}\left\Vert E(t)\right\Vert _{2}=+\infty$. 

It implies: 
\begin{equation}
\forall\alpha>0\:\exists t_{0}\in\left]t_{e}\,,\, T\right],\, x_{0}\in\mathbb{R}^{d}\, with\,\left\Vert x_{0}\right\Vert _{2}=1\, s.t\:\left|x_{0}^{T}E(t_{0})x_{0}\right|\geq\alpha\label{eq:limite_E_def}
\end{equation}

We also know it exists a unique optimal trajectory for the LQ problem
on $\left[t_{0},T\right]$ with $x(t_{0})=x_{0}$ and the associated
optimal cost is $-x_{0}^{T}E_{\theta}(t_{0})x_{0}$. But by minimality
of this cost it has to be majored by the cost \emph{$C(t_{0},0,\lambda)$}
i.e the cost associated to the control $u=0$. We can see it exists
a constant $D>0$ such $C(t_{0},0,\lambda)$ is majored by $D\left\Vert x_{0}\right\Vert _{2}^{2}$
and so: 
\[
\left|x_{0}^{T}E_{\theta}(t_{0})x_{0}\right|\leq D
\]
which contradict (\ref{eq:limite_E_def}) and finish the proof.\end{proof}
\begin{lem}
\label{lem:linear_h_wrt_X}$\forall(t,\theta)$ $h_{\theta}(t,.)$
is an affine function of \textup{$X$} and can be written under the
form:\textup{ 
\[
h_{\theta}(t,X)=N_{\theta}(t).X
\]
}With:
\[
N_{\theta}(t).X\,:=\int_{t}^{T}R_{\theta}(T-t,T-s)X(s)ds+E_{\theta}(t)X(t)+\int_{t}^{T}R_{\theta}(T-t,T-s)E_{\theta}(s)r_{\theta}(s)ds
\]
with $R_{\theta}$ defined by (\ref{eq:reverse_time_R}).\end{lem}
\begin{proof}
Considering the backward ODE:
\[
\left\{ \begin{array}{l}
\dot{h_{\theta,i}}(t,X)=\alpha_{\theta}(T-t)h_{\theta,i}(t,X)+\beta_{\theta}(T-t,X)\\
h_{\theta,i}(0,X)=0
\end{array}\right.
\]

We know thanks to Duhamel formula: 
\[
h_{\theta,i}(t,X)=\int_{0}^{t}R_{\theta}(t,s)\beta(T-s,X)ds
\]

with $R_{\theta}$ defined by \ref{eq:reverse_time_R}.

Hence: 
\[
\begin{array}{l}
h_{\theta}(t,X)=h_{i}(T-t,X)=\int_{0}^{T-t}R_{\theta}(T-t,s)\beta(T-s,X)ds\\
=\int_{t}^{T}R_{\theta}(T-t,T-s)\beta_{\theta}(s,X)ds
\end{array}
\]
Taking the value of $\beta$ and using integration by part we have:
\[
\begin{array}{lll}
h_{\theta}(t,X) & = & \int_{t}^{T}\left(R_{\theta}(T-t,T-s)E_{\theta}(s)A_{\theta}(s)+\frac{d\left(R_{\theta}(T-t,T-s)E_{\theta}(s)\right)}{ds}\right)X(s)ds\\
 & + & E_{\theta}(t)X(t)+\int_{t}^{T}R_{\theta}(T-t,T-s)E_{\theta}(s)r_{\theta}(s)ds
\end{array}
\]

And using resolvant property we finally obtain: 
\[
\frac{d\left(R_{\theta}(T-t,T-s)E_{\theta}(s)\right)}{ds}=R_{\theta}(T-t,T-s)\left(I_{p}-E_{\theta}(s)A_{\theta}(s)\right)
\]

So:
\begin{equation}
h_{\theta}(t,X)=\int_{t}^{T}R_{\theta}(T-t,T-s)X(s)ds+E_{\theta}(t)X(t)+\int_{t}^{T}R_{\theta}(T-t,T-s)E_{\theta}(s)r_{\theta}(s)ds\label{eq:h_without_der}
\end{equation}
\end{proof}
\begin{lem}
\label{lem:contolled_X_der_h_by_X_h} Under conditions 1 and 2, $\forall X\in H^{1}(\left[0,\, T\right],\mathbb{R}^{d})$
with $X(0)=x_{0}^{*}$ we have 
\[
\begin{array}{ccc}
\int_{0}^{T}\dot{X}(t)^{T}h_{\theta}(t,X)dt & = & F_{1,\theta}(X)+F_{2,\theta}(X)+F_{3,\theta}(X)\\
 & - & x_{0}^{*T}\int_{0}^{T}R_{\theta}(T,T-s)E_{\theta}(s)r_{\theta}(s)ds\\
 & - & \frac{1}{2}x_{0}^{*T}E_{\theta}(0)x_{0}^{*}
\end{array}
\]

with: $\left\{ \begin{array}{l}
F_{1,\theta}(X)=-x_{0}^{*T}\int_{0}^{T}R_{\theta}(T,T-s)X(s)ds\\
F_{2,\theta}(X)=\int_{0}^{T}X(t)^{T}\left(\alpha_{\theta}(t)h_{\theta}(t,X)+A_{\theta}(t)X+r_{\theta}(t)\right)dt\\
F_{3,\theta}(X)=\frac{1}{2}\int_{0}^{T}X(t)^{T}\dot{E_{\theta}(t)}X(t)dt
\end{array}\right.$\end{lem}
\begin{proof}
Integration by part give us:

\[
\begin{array}{l}
\int_{0}^{T}\dot{X}(t)^{T}h_{\theta}(t,X)dt=\left[X(t)^{T}h_{\theta}(t,X)\right]_{0}^{T}+\int_{0}^{T}X(t)^{T}\left(\alpha_{\theta}(t)h_{\theta}(t,X)+\beta_{\theta}(t,X)\right)dt\\
=-x_{0}^{*T}h_{\theta}(0,X)+\int_{0}^{T}X(t)^{T}\left(\alpha_{\theta}(t)h_{\theta}(t,X)+E_{\theta}(t)\left(A_{\theta}(t)X+r_{\theta}(t)\right)\right)dt-\int_{0}^{T}X(t)^{T}E_{\theta}(t)\dot{X}(t)dt
\end{array}
\]

and 
\[
\begin{array}{l}
\int_{0}^{T}X(t)^{T}E_{\theta}(t)\dot{X}(t)dt=-\frac{1}{2}\left(x_{0}^{*T}E_{\theta}(0)x_{0}^{*}+\int_{0}^{T}X(t)^{T}\dot{E_{\theta}(t)}X(t)dt\right)\end{array}
\]

Moreover using affine nature of $h$ w.r.t $X$ and using the same
notation as in proposition \lemref{linear_h_wrt_X}: 

\[
\begin{array}{ccc}
x_{0}^{*T}h_{\theta}(0,X) & = & x_{0}^{*T}\int_{0}^{T}R_{\theta}(T,T-s)X(s)ds+x_{0}^{*T}E_{\theta}(0)x_{0}^{*}\\
 & + & x_{0}^{*T}\int_{0}^{T}R_{\theta}(T,T-s)E_{\theta}(s)r_{\theta}(s)ds
\end{array}
\]

Finally we obtain:
\[
\begin{array}{lll}
\int_{0}^{T}\dot{X}(t)^{T}h_{\theta}(t,X)dt & = & -x_{0}^{*T}\int_{0}^{T}R_{\theta}(T,T-s)X(s)ds-\frac{1}{2}x_{0}^{*T}E_{\theta}(0)x_{0}^{*}\\
 & + & \int_{0}^{T}X(t)^{T}\left(\alpha_{\theta}(t)h_{\theta}(t,X)dt+\left(A_{\theta}(t)X+r_{\theta}(t)\right)\right)dt\\
 & + & \frac{1}{2}\int_{0}^{T}X(t)^{T}\dot{E_{\theta}(t)}X(t)dt\\
 & - & x_{0}^{*T}\int_{0}^{T}R_{\theta}(T,T-s)E_{\theta}(s)r_{\theta}(s)ds
\end{array}
\]

\end{proof}

\subsection{Consistency Proof}

In the following proposition \ref{prop:controlled_S_by_h_X} we show
$\left|S(\widehat{X};\theta,\lambda)-S(X^{*};\theta,\lambda)\right|$
is controlled by the distance between $\widehat{X}$ and $X^{*}$
and between $\widehat{h}$ and $h^{*}$. In proposition \ref{prop:controlled_h_by_X}
we show $\left\Vert \widehat{h_{\theta}}-h_{\theta}^{*}\right\Vert _{L^{2}}$
is uniquely controlled by $\left\Vert \widehat{X}-X^{*}\right\Vert _{L^{2}}$
the same will follow for $\left|S_{\lambda}(\theta)-S_{\lambda}^{*}(\theta)\right|$
\begin{prop}
\label{prop:controlled_S_by_h_X} Under conditions 1 and 3, $\forall\theta\in\Theta$
we have:
\[
\begin{array}{l}
\left|S(\widehat{X};\theta,\lambda)-S(X^{*};\theta,\lambda)\right|\\
\leq2\left(\bar{A}\bar{h}+K_{1}+K_{2}\left\Vert \widehat{h}_{\theta}\right\Vert _{L^{2}}+K_{3}\left\Vert \widehat{X}\right\Vert _{L^{2}}\right)\left\Vert X^{*}-\widehat{X}\right\Vert _{L^{2}}\\
+\left(\bar{A}\left\Vert \widehat{X}\right\Vert _{L^{2}}+K_{4}+\frac{1}{\lambda}\left(\left\Vert \widehat{h}_{\theta}\right\Vert _{L^{2}}+\bar{h}\right)\right)\left\Vert h_{\theta}^{*}-\widehat{h}_{\theta}\right\Vert _{L^{2}}
\end{array}
\]
With \textup{$:\left\{ \begin{array}{l}
K_{1}=\sqrt{d}\left\Vert X_{0}\right\Vert _{2}\bar{R}+d\bar{E}\bar{A}\bar{X}+\sqrt{d}\bar{\dot{E}}\overline{X}\\
K_{2}=\sqrt{d}\left(\bar{A}+\frac{\bar{E}}{\lambda}\right)\\
K_{3}=d\bar{E}\bar{A}+\sqrt{d}\bar{\dot{E}}\\
K_{4}=\sqrt{d}\left(\bar{A}+\frac{\bar{E}}{\lambda}\right)\bar{X}
\end{array}\right.$}

and: $\begin{array}{l}
\bar{R}=\sup_{\theta\in\Theta}\left\Vert R_{\theta}(T,T-.)\right\Vert _{L^{2}}\\
\bar{\dot{E}}=\sup_{\theta\in\Theta}\left\Vert \dot{E}_{\theta}\right\Vert _{L^{2}}
\end{array}$\end{prop}
\begin{proof}
For the sake of notation we will consider the homogenous case i.e
$r_{\theta}(t)=0$ 

By triangular inequality we have:

\[
\begin{array}{l}
\left|S(\widehat{X};\theta,\lambda)-S(X^{*};\theta,\lambda)\right|\\
\leq2\left|\int_{0}^{T}\left(h_{\theta}^{*}(t)^{T}A_{\theta}(t)X^{*}(t)-\widehat{h}_{\theta}(t)^{T}A_{\theta}(t)\widehat{X}(t)\right)dt\right|\\
+2\left|\int_{0}^{T}\left(\dot{\widehat{X}}(t)^{T}\widehat{h}_{\theta}(t)-\dot{X}^{*}(t)^{T}h_{\theta}^{*}(t)\right)dt\right|\\
+\frac{1}{\lambda}\left|\int_{0}^{T}\left(h_{\theta}^{*}(t)^{T}h_{\theta}^{*}(t)-\widehat{h}_{\theta}(t)^{T}\widehat{h}_{\theta}(t)\right)dt\right|
\end{array}
\]

Now we will separatly bound each of the three previous terms:

The first one:

\[
\begin{array}{l}
\left|\int_{0}^{T}\left(h_{\theta}^{*}(t)^{T}A_{\theta}(t)X^{*}(t)-\widehat{h}_{\theta}(t)^{T}A_{\theta}(t)\widehat{X}(t)\right)dt\right|\\
\leq\left|\int_{0}^{T}h_{\theta}^{*}(t)^{T}A_{\theta}(t)\left(X^{*}(t)-\widehat{X}(t)\right)dt\right|+\left|\int_{0}^{T}\left(h_{\theta}^{*}(t)-\widehat{h_{\theta}}(t)\right)^{T}A_{\theta}(t)\widehat{X}(t)dt\right|\\
\leq\left\Vert h_{\theta}^{*T}A_{\theta}\right\Vert _{L^{2}}\left\Vert X^{*}-\widehat{X}\right\Vert _{L^{2}}+\left\Vert A_{\theta}\widehat{X}\right\Vert _{L^{2}}\left\Vert h_{\theta}^{*}-\widehat{h_{\theta}}\right\Vert _{L^{2}}
\end{array}
\]

The last inequality has been obtained thanks to Cauchy-Schwarz inequality.

The second one inequality is a bit cumbersome in terms of computation.
For the sake of clarity we left some computational details in \lemref{contolled_X_der_h_by_X_h}
and we obtain with the same notation:
\[
\begin{array}{ccc}
\int_{0}^{T}\dot{\widehat{X}}(t)^{T}\widehat{h_{\theta}}(t)dt & = & F_{1,\theta}(\widehat{X})+F_{2,\theta}(\widehat{X})+F_{3,\theta}(\widehat{X})\\
 & - & x_{0}^{*T}E_{\theta}(0)x_{0}^{*}
\end{array}
\]

and:

\[
\begin{array}{ccc}
\int_{0}^{T}\dot{X}(t)^{*T}h_{\theta}^{*}(t)dt & = & F_{1,\theta}(X^{*})+F_{2,\theta}(X^{*})+F_{3,\theta}(X^{*})\\
 & - & x_{0}^{*T}E_{\theta}(0)x_{0}^{*}
\end{array}
\]

Hence we can formulate $S(\widehat{X};\theta,\lambda)$ without the
derivative form expression and the last decomposition allows us to
bound $\left|\int_{0}^{T}\left(\dot{\widehat{X}}(t)^{T}\widehat{h_{\theta}}(t)-\dot{X}^{*}(t)^{T}h_{\theta}^{*}(t)\right)dt\right|$
only with $\left\Vert \widehat{X}-X^{*}\right\Vert _{L^{2}}$ and
$\left\Vert \widehat{h}_{\theta}-h_{\theta}^{*}\right\Vert _{L^{2}}$

By use of norm inequalities we obtain the following bounds: 

\[
\begin{array}{l}
\left|F_{1,\theta}(\widehat{X})-F_{1,\theta}(X^{*})\right|\leq\sqrt{d}\left\Vert X_{0}\right\Vert _{2}\bar{R}\left\Vert \widehat{X}-X^{*}\right\Vert _{L^{2}}\\
\begin{array}{lll}
\left|F_{2,\theta}(\widehat{X})-F_{2,\theta}(X^{*})\right| & \leq & \sqrt{d}\left(\bar{A}+\frac{\bar{E}}{\lambda}\right)\left(\left\Vert \widehat{X}-X^{*}\right\Vert _{L^{2}}\left\Vert \widehat{h_{\theta}}\right\Vert _{L^{2}}+\bar{X}\left\Vert \widehat{h}_{\theta}-h_{\theta}^{*}\right\Vert _{L^{2}}\right)\\
 & + & \sqrt{d}\bar{A}\left(\left\Vert \widehat{X}\right\Vert _{L^{2}}+\bar{X}\right)\left\Vert \widehat{X}-X^{*}\right\Vert _{L^{2}}
\end{array}\\
\left|F_{3,\theta}(\widehat{X})-F_{3,\theta}(X^{*})\right|\leq\sqrt{d}\left\Vert \widehat{X}-X^{*}\right\Vert _{L^{2}}\bar{\dot{E}}\left(\left\Vert \widehat{X}\right\Vert _{L^{2}}+\overline{X}\right)
\end{array}
\]

And we obtain for the second part:

\[
\begin{array}{lll}
\left|\int_{0}^{T}\left(\dot{\widehat{X}}(t)^{T}\widehat{h}_{\theta}(t)-\dot{X}(t)^{*T}h_{\theta}^{*}(t)\right)dt\right| & \leq & \left(K_{1}+K_{2}\left\Vert \widehat{h_{\theta}}\right\Vert _{L^{2}}+K_{3}\left\Vert \widehat{X}\right\Vert _{L^{2}}\right)\left\Vert \widehat{X}-X^{*}\right\Vert _{L^{2}}\\
 &  & +K_{4}\left\Vert \widehat{h_{\theta}}-h_{\theta}^{*}\right\Vert _{L^{2}}
\end{array}
\]

$\textrm{With: }\left\{ \begin{array}{l}
K_{1}=\sqrt{d}\left\Vert X_{0}\right\Vert _{2}\bar{R}+\sqrt{d}\bar{A}\bar{X}+\sqrt{d}\bar{\dot{E}}\overline{X}\\
K_{2}=\sqrt{d}\left(\bar{A}+\frac{\bar{E}}{\lambda}\right)\\
K_{3}=\sqrt{d}\bar{A}+\sqrt{d}\bar{\dot{E}}\\
K_{4}=\sqrt{d}\left(\bar{A}+\frac{\bar{E}}{\lambda}\right)\bar{X}
\end{array}\right.$

For the third one we have:

\[
\begin{array}{l}
\left|\int_{0}^{T}\left(h_{\theta}^{*}(t)^{T}h_{\theta}^{*}(t)-\widehat{h_{\theta}}(t)^{T}\widehat{h_{\theta}}(t)\right)dt\right|\\
=\left|\int_{0}^{T}\left(h_{\theta}^{*}(t)^{T}\left(h_{\theta}^{*}(t)-\widehat{h_{\theta}}(t)\right)-\widehat{h_{\theta}}(t)^{T}\left(\widehat{h_{\theta}}(t)-h_{\theta}^{*}(t)\right)\right)dt\right|\\
\leq\left(\left\Vert \widehat{h_{\theta}}\right\Vert _{L^{2}}+\left\Vert h_{\theta}^{*}\right\Vert _{L^{2}}\right)\left\Vert h_{\theta}^{*}-h_{\theta}\right\Vert _{L^{2}}
\end{array}
\]

Hence by summing we finish the proof.\end{proof}
\begin{prop}
\label{prop:controlled_h_by_X} Under conditions 1 and 3 $\forall\theta\in\Theta$
we have: \textup{
\[
\begin{array}{l}
\left\Vert \widehat{h_{\theta}}-h_{\theta}^{*}\right\Vert _{L^{2}}\leq K_{5}\left\Vert \widehat{X}-X^{*}\right\Vert _{L^{2}}\\
\textrm{with : }K_{5}=\sqrt{d}\left(Tde^{\sqrt{d}\left(\overline{A}+\frac{\overline{E}}{\lambda}\right)T}+\overline{E}\right)
\end{array}
\]
 }\end{prop}
\begin{proof}
Thanks lemma \ref{lem:linear_h_wrt_X} we have the following affine
dependance of $h$ w.r.t $X$: 
\[
\widehat{h_{\theta}}(t)-h_{\theta}^{*}(t)=\int_{t}^{T}R_{\theta}(T-t,T-s)\left(\widehat{X}(s)-X^{*}(s)\right)ds+E_{\theta}(t)\left(\widehat{X}(t)-X^{*}(t)\right)
\]

Taking the norm gives us: 
\[
\begin{array}{lll}
\left\Vert \widehat{h}_{\theta}(t)-h_{\theta}^{*}(t)\right\Vert _{2} & \leq & \left\Vert \int_{t}^{T}R_{\theta}(T-t,T-s)\left(\widehat{X}(s)-X^{*}(s)\right)ds\right\Vert _{2}\\
 & + & \left\Vert E_{\theta}(t)\left(\widehat{X}(t)-X^{*}(t)\right)\right\Vert _{2}\\
 & \leq & \sqrt{d}\left(\sqrt{T}de^{\sqrt{d}\left(\overline{A}+\frac{\overline{E}}{\lambda}\right)T}\left\Vert \widehat{X}-X^{*}\right\Vert _{L^{2}}+\left\Vert E_{\theta}(t)\right\Vert _{2}\left\Vert \widehat{X}(t)-X^{*}(t)\right\Vert _{2}\right)
\end{array}
\]

Using condition C1 and C3 and the upper bound $\left\Vert R_{\theta}(T-t,T-s)\right\Vert _{2}\leq de^{\sqrt{d}\left(\overline{A}+\frac{\overline{E}}{\lambda}\right)T}$
thanks to proposition 3 in supplementary material.

Finally we obtain: 
\[
\left\Vert \widehat{h_{\theta}}-h_{\theta}^{*}\right\Vert _{L^{2}}\leq\sqrt{d}\left(Tde^{\sqrt{d}\left(\overline{A}+\frac{\overline{E}}{\lambda}\right)T}+\left\Vert E_{\theta}\right\Vert _{L^{2}}\right)\left\Vert \widehat{X}-X^{*}\right\Vert _{L^{2}}
\]

\end{proof}

\subsection{Asymptotic normality proof}

The proof of continuity of some functionals useful for proposition
\propref{linear_decomposition} are left in the supplementary materials,
as they require cumbersome computations and they does not provide
particular insights in the mechanics of the proofs.
\begin{prop}
\label{prop:linear_decomposition_ap}Under conditions 1-5, we have
: 
\[
\widehat{\theta}^{T}-\theta^{*}=2\frac{\partial^{2}S(X^{*};\theta^{*},\lambda)}{\partial\theta^{T}\partial\theta}^{-1}\left(\Gamma(\widehat{X})-\Gamma(X^{*})\right)+o_{P}(1)
\]
where $\Gamma\::\: C\left(\left[0\,,\, T\right],\mathbb{R}^{d}\right)\:\rightarrow\mathbb{R}^{p}$
is a linear functional defined by 
\begin{equation}
\Gamma(X)=\int_{0}^{T}\left(\frac{\partial\left(A_{\theta*}(t).X^{*}\right)}{\partial\theta}+\frac{1}{\lambda}\frac{\partial h_{\theta^{*}}(t,X^{*})}{\partial\theta}\right)^{T}\left(\int_{t}^{T}R_{\theta^{*}}(T-t,T-s)X(s)ds\right)dt.\label{eq:DefinitionGamma_ap}
\end{equation}
 $R_{\theta^{*}}$ is defined \textup{by (\ref{eq:reverse_time_R}).}\end{prop}
\begin{proof}
For the sake of notational simplicity here $\widehat{\theta}^{T}$
will be simply denoted $\widehat{\theta}$. 

The first order optimal condition is 
\[
\nabla_{\theta}S(\hat{X};\widehat{\theta},\lambda)=0
\]
Equivalently, we have 
\begin{equation}
\int_{0}^{T}\frac{\partial\left(A_{\widehat{\theta}}(t).\widehat{X}+r_{\widehat{\theta}}(t)\right)}{\partial\theta}^{T}h_{\widehat{\theta}}(t,\widehat{X})+\frac{\partial h_{\widehat{\theta}}(t,\widehat{X})}{\partial\theta}^{T}\left(A_{\widehat{\theta}}(t).\widehat{X}+r_{\widehat{\theta}}(t)-\dot{\widehat{X}}\right)+\frac{1}{\lambda}\frac{\partial h_{\widehat{\theta}}(t,\widehat{X})}{\partial\theta}^{T}h_{\widehat{\theta}}(t,\widehat{X})=0\label{eq:first_order_condition_ap}
\end{equation}
We use the following decomposition for $A_{\widehat{\theta}}(t).\widehat{X}-\dot{\widehat{X}}$
and $h_{\widehat{\theta}}(t,\widehat{X})$: 
\[
\begin{array}{l}
A_{\widehat{\theta}}(t).\widehat{X}+r_{\widehat{\theta}}(t)-\dot{\widehat{X}}=A_{\widehat{\theta}}(t)\left(\widehat{X}-X^{*}\right)+\frac{\partial\left(A_{\widetilde{\theta}}(t).X^{*}+r_{\widetilde{\theta}}(t)\right)}{\partial\theta}\left(\widehat{\theta}-\theta^{*}\right)+\left(\dot{X}^{*}-\dot{\widehat{X}}\right)\\
h_{\widehat{\theta}}(t,\widehat{X})=\frac{\partial\left(h_{\widetilde{\theta}}(t,\widehat{X})\right)}{\partial\theta}\left(\widehat{\theta}-\theta^{*}\right)+N_{\theta^{*}}(t).\left(\widehat{X}-X^{*}\right)
\end{array}
\]
with $\widetilde{\theta}$ being a random point between $\theta^{*}$
and $\widehat{\theta}$ and $N$ defined as in \lemref{linear_h_wrt_X}.
By replacing in (\ref{eq:first_order_condition_ap}), we obtain:

\begin{equation}
\int_{0}^{T}H_{1}(t,\widehat{\theta},\widehat{X})dt\;\left(\widehat{\theta}-\theta^{*}\right)=\int_{0}^{T}H_{2}(t,\widehat{\theta},\widehat{X})\left(\widehat{X}-X^{*}\right)-\frac{\partial\left(h_{\widehat{\theta}}(t,\widehat{X})\right)}{\partial\theta}^{T}\left(\dot{X}^{*}-\dot{\widehat{X}}\right)dt\label{eq:linear_form_non_as_ap}
\end{equation}
with 
\[
\begin{array}{l}
H_{1}(t,\widehat{\theta},\widehat{X})=\frac{\partial\left(A_{\widehat{\theta}}(t).\widehat{X}+r_{\widehat{\theta}}(t)\right)}{\partial\theta}^{T}\frac{\partial h_{\widetilde{\theta}}(t,\widehat{X})}{\partial\theta}+\frac{\partial\left(h_{\widehat{\theta}}(t,\widehat{X})\right)^{T}}{\partial\theta}\frac{\partial\left(A_{\widetilde{\theta}}(t).X^{*}+r_{\widetilde{\theta}}(t)\right)}{\partial\theta}+\frac{1}{\lambda}\frac{\partial\left(h_{\widehat{\theta}}(t,\widehat{X})\right)}{\partial\theta}^{T}\frac{\partial h_{\widetilde{\theta}}(t,\widehat{X})}{\partial\theta}\\
H_{2}(t,\widehat{\theta},\widehat{X})=\frac{\partial\left(A_{\widehat{\theta}}(t).\widehat{X}+r_{\widehat{\theta}}(t)\right)}{\partial\theta}^{T}N_{\theta^{*}}(t)+\frac{\partial\left(h_{\widehat{\theta}}(t,\widehat{X})\right)^{T}}{\partial\theta}A_{\widehat{\theta}}(t)+\frac{1}{\lambda}\frac{\partial\left(h_{\widehat{\theta}}(t,\widehat{X})\right)}{\partial\theta}^{T}N_{\theta^{*}}(t)
\end{array}
\]
Thanks to propositions in supplementary material, the following functionals
\[
\left\{ \begin{array}{l}
D_{1}:\theta\longmapsto\left(t\longmapsto A_{\theta}(t)\right)\\
D_{2}:\left(\theta,X\right)\longmapsto\left(t\longmapsto\frac{\partial\left(A_{\theta}(t).X\right)}{\partial\theta}\right)\\
D_{3}:\left(\theta,X\right)\longmapsto\left(t\longmapsto h_{\theta}(t,X)\right)\\
D_{4}:\left(\theta,X\right)\longmapsto\left(t\longmapsto\frac{\partial\left(h_{\theta}(t,X)\right)}{\partial\theta}\right)
\end{array}\right.
\]
are continuous on $\Theta\times L^{2}\left(\left[0\,,\, T\right],\mathbb{R}^{d}\right)$,
and the continuous mapping theorem implies that $t\longmapsto H_{1}(t,\widehat{\theta},\widehat{X})$
and $t\longmapsto H_{2}(t,\widehat{\theta},\widehat{X})$ converge
in probability in the $L^{2}$ sense to the function $t\longmapsto H_{1}(t,\theta^{*},X^{*})$
and $t\longmapsto H_{2}(t,\theta^{*},X^{*})$. So $\left\Vert H_{1}(.,\widehat{\theta},\widehat{X})\right\Vert _{L^{2}}$
converges in probability to $\left\Vert H_{1}(.,\theta^{*},X^{*})\right\Vert _{L^{2}}$
and so it is bounded. Finally, we have the convergence in probability
of each entry of $\int_{0}^{T}H_{1}(t,\widehat{\theta},\widehat{X})dt$
to the corresponding entry to $\int_{0}^{T}H_{1}(t,\theta^{*},X^{*})dt$.
Moreover, condition C5 assumes that the Hessian 
\[
\int_{0}^{T}H_{1}(t,\theta^{*},X^{*})dt=\frac{1}{2}\frac{\partial^{2}S(X^{*};\theta^{*},\lambda)}{\partial\theta^{T}\partial\theta}
\]
is nonsingular at $\theta=\theta^{*}$. Finally, we have 
\[
\int_{0}^{T}H_{1}(t,\widehat{\theta},\widehat{X})dt\overset{P}{\longrightarrow}\frac{1}{2}\frac{\partial^{2}S(X^{*};\theta^{*},\lambda)}{\partial\theta^{T}\partial\theta}
\]
By an analogous reasoning, the asymptotic behavior of $\widehat{\theta}-\theta^{*}$
is given by 
\[
2\frac{\partial^{2}S(X^{*};\theta^{*},\lambda)}{\partial\theta^{T}\partial\theta}^{-1}\left(\int_{0}^{T}H_{2}(t,\theta^{*},X^{*})\left(\widehat{X}-X^{*}\right)dt-\frac{\partial\left(h_{\theta^{*}}(t,X^{*})\right)}{\partial\theta}^{T}\left(\dot{X}^{*}-\dot{\widehat{X}}\right)dt\right)
\]
and Integration By Part gives
\[
\begin{array}{ccc}
\int_{0}^{T}\frac{\partial h_{\theta^{*}}(t,X^{*})}{\partial\theta}^{T}\left(\dot{X}^{*}-\dot{\widehat{X}}\right)dt & = & \left[\frac{\partial h_{\theta^{*}}(t,X^{*})}{\partial\theta}^{T}\left(X^{*}-\widehat{X}\right)\right]_{0}^{T}\\
 & - & \int_{0}^{T}\frac{d}{dt}\left(\frac{\partial h_{\theta^{*}}(t,X^{*})}{\partial\theta}^{T}\right)\left(X^{*}-\widehat{X}\right)dt
\end{array}
\]
But, as $\frac{\partial h(T,\theta^{*},X^{*})}{\partial\theta}=0$
and $\widehat{X}(0)=x_{0}^{*}$ we have:

\[
\begin{array}{l}
\int_{0}^{T}\left(H_{2}(t,\theta^{*},X^{*})+\frac{d}{dt}\left(\frac{\partial h_{\theta^{*}}(t,X^{*})}{\partial\theta}^{T}\right)\right).X(t)dt\\
=\int_{0}^{T}\left(\frac{\partial\left(A_{\theta*}(t).X^{*}+r_{\theta^{*}}(t)\right)}{\partial\theta}+\frac{1}{\lambda}\frac{\partial h_{\theta^{*}}(t,X^{*})}{\partial\theta}\right)^{T}\left(h_{\theta^{*}}(t,X)-E_{\theta^{*}}(t).X(t)\right)dt
\end{array}
\]
Hence we can write
\[
\widehat{\theta}-\theta^{*}=2\frac{\partial^{2}S(X^{*};\theta^{*},\lambda)}{\partial\theta^{T}\partial\theta}^{-1}\left(\Gamma(\widehat{X})-\Gamma(X^{*})\right)+o_{P}(1)
\]
with 
\[
\Gamma(X)=\int_{0}^{T}\left(\frac{\partial\left(A_{\theta*}(t).X^{*}\right)}{\partial\theta}+\frac{1}{\lambda}\frac{\partial h_{\theta^{*}}(t,X^{*})}{\partial\theta}\right)^{T}\left(\int_{t}^{T}R_{\theta^{*}}(T-t,T-s)X(s)ds\right)dt
\]
where $R_{\theta^{*}}$ is defined by (\ref{eq:reverse_time_R}).\end{proof}
\begin{prop}
\label{prop:newey_result_ap}Under conditions 1-8 and by defining
$\Gamma$ as in proposition \propref{linear_decomposition_ap} we
have that $\Gamma(\widehat{X})-\Gamma(X^{*})$ is asymptotically normal
and $\Gamma(\widehat{X})-\Gamma(X^{*})=O_{P}(n^{-1/2})$\end{prop}
\begin{proof}
This proposition is a direct consequence of Theorem 9 in \cite{Newey1997}.
The conditions to be satisfied are 
\begin{enumerate}
\item $\left(Y_{i},t_{i}\right)$ are i.i.d with $Var(Y\mid t)$ bounded.
\item $E(\left(Y-X^{*}(t)\right)^{4}\mid t)$ is bounded, and $Var(Y\mid t)$
is bounded away from 0.
\item The support of $t$ is a compact interval on which $t$ has a probability
density function bounded away from 0.
\item There is $v(t)$ such that $E(v(t)v(t)^{T})$ is finite and non-singular
such that: $D(\Gamma)(X^{*})(X^{*})=E(v(t)X^{*}(t))$ and $D(\Gamma)(X^{*})(p_{kK})=E(v(t)p_{kK}(t))$
for all $k$ and $K$ and there is $\beta_{K}$ with $E(\left\Vert v(t)-\beta_{K}p_{K}(t)\right\Vert _{2}^{2})\rightarrow0$
\item $X^{*}(t)=E(Y\mid t)$ is derivable of order $s$ on the support of
$t$.
\end{enumerate}
Requirements 1,2,3 are direct consequences of conditions C6 and C7
(and the solution is always defined on $\left[0,T\right]$). 

For the fourth requirement we will consider the monodimensional case
$d=1$. We know that $\Gamma$ is linear and continuous on $L^{2}\left(\left[0,\, T\right],\mathbb{R}^{d}\right)$
thanks to conditions C1 and C3-4 and hence differentiable with: $D(\Gamma)(X^{*})(X)=\Gamma(X)$
. By the Riesz-Frechet representation theorem we have: $v\in L^{2}(\left[0\,,\, T\right],\mathbb{R})$
s.t $\Gamma(X)=\int_{0}^{T}v(t)X(t)dt$ which verify the three conditions
of the forth requirement. Starting from the mono-dimensional case,
multi-dimensional case can be made componentwise.

Requirement 5 is a simple consequence of the condition C8.
\end{proof}
\newpage{}

\section{Gradient Computation : Adjoint Method \& Sensitivity equation}

\subsection{Notation and partial derivative computation}

For optimization purpose we need to compute the gradient of $S(\widehat{X};\theta,\lambda)$.
For this we will present two methods: a direct approach using sensitivity
equation and a second one using adjoint method.

\subsubsection{Row vector notation for the vector field of the general Riccati equation}

We will define the solution of the general Riccati equation in row
formulation, we introduce 
\[
Q_{\theta}(t)=\left(\widehat{h_{\theta}}^{T},\left(E_{\theta}^{r}\right)^{T}\right)^{T}(t)
\]
with $E_{\theta}^{r}:=\left(E_{\theta,1}^{T},\cdots,E_{\theta,d}^{T}\right)^{T}$the
row formulation of $E_{\theta}$, $E_{\theta,i}$ beeing the $i-th$
column of $E_{\theta}$. It is a $D:=d^{2}+d$ sized function respecting
the ODE :
\[
\begin{array}{l}
\dot{Q_{\theta}}=F(Q_{\theta},\theta,t)\\
Q_{\theta}(T)=0
\end{array}
\]

by introducing the general vector field $F$:

\[
F(Q_{\theta},\theta,t)=\left(\begin{array}{c}
G(Q_{\theta},\theta,t)\\
H(Q_{\theta},\theta)
\end{array}\right)
\]

with $G$ and $H$ defined by: 
\[
\begin{array}{lll}
G(Q_{\theta},\theta,t) & := & -\left(A_{\theta}(t)^{T}+\frac{E_{\theta}}{\lambda}\right)\widehat{h_{\theta}}-E_{\theta}\left(A_{\theta}(t)\widehat{X}(t)-\dot{\widehat{X}}(t)+r_{\theta}(t)\right)\\
H_{(j-1)d+i}(Q_{\theta},\theta) & := & \delta_{i,j}-(A_{\theta,i}^{T}E_{j}+A_{\theta,j}^{T}E_{\theta,i}+\frac{1}{\lambda}E_{\theta,i}^{T}E_{\theta,j})
\end{array}
\]

and $A_{\theta,i}$ beeing the $i-th$ column of $A_{\theta}$.

We also introduce:
\[
g(Q_{\theta},\theta,t)=-2\left(A_{\theta}(t)\widehat{X}(t)-\dot{\widehat{X}}(t)+r_{\theta}(t)\right)^{T}\widehat{h_{\theta}}-\frac{1}{\lambda}\widehat{h_{\theta}}^{T}\widehat{h_{\theta}}
\]

In order to write our system under the row form: 

\begin{equation}
\begin{array}{l}
S(\widehat{X};\theta,\lambda)\,:=\int_{0}^{T}g(Q_{\theta}(t),\theta,t)dt\\
\left\{ \begin{array}{l}
\dot{Q_{\theta}}=F(Q_{\theta},\theta,t)\\
Q_{\theta}(T)=0
\end{array}\right.
\end{array}\label{eq:func_criteria_simplified}
\end{equation}
For the next subsections we will drop dependence in $\theta$ for
$A_{\theta},\, r_{\theta},\, E_{\theta},\,\widehat{h_{\theta}}$

\subsubsection{Partial derivative of Riccati vector field}

In order to compute sensitivity equation or adjoint model we need
to compute $\frac{\partial g}{\partial\theta}(Q_{\theta},\theta,t),\frac{\partial g}{\partial Q}(Q_{\theta},\theta,t),$
$\frac{\partial F}{\partial\theta}(Q_{\theta},\theta,t)$ and$\frac{\partial F}{\partial Q}(Q_{\theta},\theta,t)$

The computation for $\frac{\partial g}{\partial\theta}(Q_{\theta},\theta,t),\:\frac{\partial g}{\partial Q}(Q_{\theta},\theta,t)$
is straightforward

\[
\frac{\partial g}{\partial\theta}(Q_{\theta},\theta,t)=-2\widehat{h}^{T}\left(\frac{\partial\left(A(t)\widehat{X}(t)\right)}{\partial\theta}+\frac{\partial r}{\partial\theta}(t)\right)
\]

\[
\frac{\partial g}{\partial Q}(Q_{\theta},\theta,t)=\left(-2\left(A(t)\widehat{X}(t)-\dot{\widehat{X}}(t)+r(t)+\frac{\widehat{h}}{\lambda}\right)^{T},0_{1,d^{2}}\right)
\]

For $\frac{\partial F}{\partial\theta}(h,E^{r},\theta,t)$ and $\frac{\partial F}{\partial Q}(R_{\theta},\theta,t)$
we obtain

\[
\begin{array}{l}
\frac{\partial F}{\partial\theta}(Q_{\theta},\theta,t)=\left(\begin{array}{c}
\frac{\partial G}{\partial\theta}(Q_{\theta},\theta,t)\\
\frac{\partial H}{\partial\theta}(Q_{\theta},\theta)
\end{array}\right)\\
\frac{\partial F}{\partial Q}(Q_{\theta},\theta,t)=\left(\begin{array}{cc}
-\left(A(t)^{T}+\frac{E}{\lambda}\right) & \frac{\partial G_{i}}{\partial E_{j}^{r}}(Q_{\theta},\theta,t)\\
0_{d^{2},d} & \frac{\partial H(Q_{\theta},\theta)}{\partial E^{r}}
\end{array}\right)
\end{array}
\]

with:
\[
\begin{array}{l}
\frac{\partial G_{i}}{\partial E_{(k-1)d+h}^{r}}(Q_{\theta},\theta,t)=-\delta_{i,h}\left(\frac{\widehat{h}}{\lambda}+A(t)\widehat{X}(t)-\dot{\widehat{X}}(t)+r(t)\right)_{k}\\
\frac{\partial G}{\partial\theta}(Q_{\theta},\theta,t)=-\left(h^{T}\frac{\partial A_{i}(t)}{\partial\theta}\right)_{1\leq i\leq d}-E\left(\frac{\partial\left(A(t)\widehat{X}(t)\right)}{\partial\theta}+\frac{\partial r(t)}{\partial\theta}\right)
\end{array}
\]

We also need to compute $H(Q_{\theta},\theta)$ partial derivative
w.r.t $E^{r}$ and $\theta$.

We have:
\[
\left(\frac{\partial H(Q_{\theta},\theta)}{\partial E^{r}}\right)_{(j-1)d+i}=-\left(\begin{array}{ccccc}
0 & A_{j}^{t} & 0 & A_{i}^{t} & 0\end{array}\right)-\frac{1}{\lambda}\left(\begin{array}{ccccc}
0 & E_{j}^{t} & 0 & E_{i}^{t} & 0\end{array}\right)
\]

Because:
\begin{itemize}
\item $\frac{\partial}{\partial E^{r}}\left(A_{j}^{t}E_{i}+A_{i}^{t}E_{j}\right)=\left(\begin{array}{ccccc}
0 & A_{j}^{t} & 0 & A_{i}^{t} & 0\end{array}\right)$ where $A_{j}^{t}$ is in $i-th$ position and $A_{i}^{t}$ is in
$j-th$ position.
\item $\frac{1}{\lambda}\frac{\partial}{\partial E}\left(E_{j}^{t}E_{i}\right)=\left(\begin{array}{ccccc}
0 & \frac{1}{\lambda}E_{j}^{t} & 0 & 0 & 0\end{array}\right)+\left(\begin{array}{ccccc}
0 & 0 & \frac{1}{\lambda}E_{i}^{t} & 0 & 0\end{array}\right)$ where $E_{j}^{t}$ is in $i-th$ position and $E_{i}^{t}$ is in
$j-th$ position.
\end{itemize}
And:

\[
\left(\frac{\partial H(Q_{\theta},\theta)}{\partial\theta}\right)_{(j-1)d+i}=-E_{i}^{t}\frac{\partial A_{j}}{\partial\theta}-E_{j}^{t}\frac{\partial A_{i}}{\partial\theta}
\]

\begin{itemize}
\item Because $\frac{\partial}{\partial\theta}\left(A_{j}^{t}E_{i}+A_{i}^{t}E_{j}\right)=E_{i}^{t}\frac{\partial A_{j}}{\partial\theta}+E_{j}^{t}\frac{\partial A_{i}}{\partial\theta}$
where $\frac{\partial A_{i}}{\partial\theta}=\left(\frac{\partial A_{i}}{\partial\theta_{1}}\cdots\frac{\partial A_{i}}{\partial\theta_{p}}\right)$
a $d\times p$ matrix
\end{itemize}

\subsection{Gradient computation by sensitivity equation }

By Gradient definition we have 
\[
\nabla_{\theta}S(\widehat{X};\theta,\lambda)=\int_{0}^{T}\frac{\partial g(Q_{\theta}(t),\theta,t)}{\partial Q}\frac{\partial Q_{\theta}(t)}{\partial\theta}+\frac{\partial g(Q_{\theta}(t),\theta,t)}{\partial\theta}dt
\]

With $\frac{\partial Q_{\theta}(t)}{\partial\theta}$ solution of
the sensitivity equation: 
\[
\frac{d}{dt}(\frac{\partial Q_{\theta}(t)}{\partial\theta})=\frac{\partial F}{\partial Q}(Q_{\theta}(t),\theta,t)\frac{\partial Q_{\theta}(t)}{\partial\theta}+\frac{\partial F}{\partial\theta}(Q_{\theta}(t),\theta,t)
\]

And we know that $Q_{\theta}(T)=0$ so $\frac{\partial Q_{\theta}(T)}{\partial\theta}=0$,
hence we can obtain $\frac{\partial Q_{\theta}(t)}{\partial\theta}$
by solving the Cauchy problem: 

\[
\begin{array}{l}
\frac{d}{dt}(\frac{\partial Q_{\theta}(t)}{\partial\theta})=\frac{\partial F}{\partial Q}(Q_{\theta}(t),\theta,t)\frac{\partial Q_{\theta}(t)}{\partial\theta}+\frac{\partial F}{\partial\theta}(Q_{\theta}(t),\theta,t)\\
\frac{\partial Q_{\theta}(T)}{\partial\theta}=0
\end{array}
\]

\subsection{Gradient computation by adjoint Method}

Once again we have 
\[
\nabla_{\theta}S(\widehat{X};\theta,\lambda)=\int_{0}^{T}\frac{\partial g(Q_{\theta}(t),\theta,t)}{\partial Q}\frac{\partial Q_{\theta}(t)}{\partial\theta}+\frac{\partial g(Q_{\theta}(t),\theta,t)}{\partial\theta}dt
\]

With $\frac{\partial Q_{\theta}(t)}{\partial\theta}$ solution of
the sensitivity equation: 
\[
\frac{d}{dt}(\frac{\partial Q_{\theta}(t)}{\partial\theta})=\frac{\partial F}{\partial Q}(Q_{\theta}(t),\theta,t)\frac{\partial Q_{\theta}(t)}{\partial\theta}+\frac{\partial F}{\partial\theta}(Q_{\theta}(t),\theta,t)
\]
If we premultiply the right and left term of the previous ODE by the
$D-$sized adjoint vector $P(t)$ and then integrate we obtain 
\[
\int_{0}^{T}P(t).\frac{d}{dt}(\frac{\partial Q_{\theta}(t)}{\partial\theta})dt=\int_{0}^{T}P(t).\frac{\partial F}{\partial Q}(Q_{\theta}(t),\theta,t)\frac{\partial Q_{\theta}(t)}{\partial\theta}dt+\int_{0}^{T}P(t).\frac{\partial F}{\partial\theta}(Q_{\theta}(t),\theta,t)dt
\]

Integration by part gives us 
\[
\int_{0}^{T}P(t).\frac{d}{dt}(\frac{\partial Q_{\theta}(t)}{\partial\theta})dt=P(T).\frac{\partial Q_{\theta}(T)}{\partial\theta}-P(0).\frac{\partial Q_{\theta}(0)}{\partial\theta}-\int_{0}^{T}\dot{P}(t).\frac{\partial Q_{\theta}(t)}{\partial\theta}dt
\]

We already know that $\frac{\partial Q_{\theta}(T)}{\partial\theta}=0$
and if we take $P(0)=0$ we obtain the variational relation:

\[
\int_{0}^{T}\left(\dot{P}(t)+P(t).\frac{\partial F}{\partial Q}(Q_{\theta}(t),\theta,t)\right)\frac{\partial Q_{\theta}(t)}{\partial\theta}dt+\int_{0}^{T}P(t).\frac{\partial F}{\partial\theta}(Q_{\theta}(t),\theta,t)dt=0
\]

and by imposing: 
\[
\dot{P}(t)+P(t).\frac{\partial F}{\partial Q}(Q_{\theta},\theta,t)=\frac{\partial g(Q_{\theta}(t),\theta,t)}{\partial Q}
\]
we deduce that 
\[
\int_{0}^{T}\frac{\partial g(Q_{\theta}(t),\theta,t)}{\partial Q}\frac{\partial Q_{\theta}(t)}{\partial\theta}dt=-\int_{0}^{T}P(t).\frac{\partial Q}{\partial\theta}(Q_{\theta}(t),\theta,t)dt
\]
and so
\[
\nabla_{\theta}S(\widehat{X};\theta,\lambda)=\int_{0}^{T}\frac{\partial g(Q_{\theta}(t),\theta,t)}{\partial\theta}-P(t).\frac{\partial F}{\partial\theta}(Q_{\theta}(t),\theta,t)dt
\]

We propose here an alternative for gradient computation, we compute
$\nabla_{\theta}S(\widehat{X};\theta,\lambda)$ by considering: 
\[
\begin{array}{c}
\nabla_{\theta}S(\widehat{X};\theta,\lambda)=\int_{0}^{T}\frac{\partial g(Q_{\theta}(t),\theta,t)}{\partial\theta}-P(t).\frac{\partial F}{\partial\theta}(Q_{\theta}(t),\theta,t)dt\\
\begin{array}{l}
\dot{P}(t)=\frac{\partial g(Q_{\theta}(t),\theta,t)}{\partial Q}-P(t).\frac{\partial F}{\partial Q}(Q_{\theta}(t),\theta,t)\\
P(0)=0
\end{array}
\end{array}
\]

The interest here is computational, computing gradient by solving
sensitivity equation drives us to solve a $D\times p$ ODE system.
Here the adjoint system defining $P$ is only of size $D$.

\newpage{}

\section{Asymptotic variance expression}

We know asymptotically $\widehat{\theta}^{T}-\theta^{*}$ behaves
as: 
\[
2\frac{\partial^{2}S(X^{*};\theta^{*},\lambda)}{\partial\theta^{T}\partial\theta}^{-1}\left(\Gamma(\widehat{X})-\Gamma(X^{*})\right)
\]

with:
\[
\frac{1}{2}\frac{\partial^{2}S(X^{*};\theta^{*},\lambda)}{\partial\theta^{T}\partial\theta}=
\]

\[
\frac{\partial\left(A_{\theta^{*}}(t).X^{*}\right)}{\partial\theta}^{T}\frac{\partial h_{\theta^{*}}(t,X^{*})}{\partial\theta}+\frac{\partial\left(h_{\theta^{*}}(t,X^{*})\right)^{T}}{\partial\theta}\frac{\partial\left(A_{\theta^{*}}(t).X^{*}\right)}{\partial\theta}+\frac{1}{\lambda}\frac{\partial\left(h_{\theta^{*}}(t,X^{*})\right)}{\partial\theta}^{T}\frac{\partial h_{\theta^{*}}(t,X^{*})}{\partial\theta}
\]

the hessian of the asymptotic criteria at $\theta=\theta^{*}$

and: 
\[
\Gamma(X)=\int_{0}^{T}\left(\frac{\partial\left(A_{\theta*}(t).X^{*}\right)}{\partial\theta}+\frac{1}{\lambda}\frac{\partial h_{\theta^{*}}(t,X^{*})}{\partial\theta}\right)^{T}\left(\int_{t}^{T}R_{\theta^{*}}(T-t,T-s)X(s)ds\right)dt
\]

A linear functional w.r.t to $X$ so asymptotically: 

\[
Var(\widehat{\theta}^{T})=4\frac{\partial^{2}S(X^{*};\theta^{*},\lambda)}{\partial\theta^{T}\partial\theta}^{-1}Var(\Gamma(\widehat{X}))\frac{\partial^{2}S(X^{*};\theta^{*},\lambda)}{\partial\theta^{T}\partial\theta}^{-1}
\]

If $\widehat{X}$ is a b-Splines basis decomposition estimator under
the form $\widehat{X}=\sum_{i=1}^{K}\widehat{\beta}_{iK}p_{iK}(t)$
we can formulate $\Gamma$ as a linear function w.r.t coefficients
$\widehat{\beta}_{iK}$: 
\[
\Gamma(\widehat{X}):=P(\theta^{*},X^{*})\widehat{\beta}_{K}
\]

with: 
\[
P_{i}(\theta,X)=\int_{0}^{T}\left(\frac{\partial\left(A_{\theta}(t).X\right)}{\partial\theta}+\frac{1}{\lambda}\frac{\partial h_{\theta}(t,X)}{\partial\theta}\right)^{T}\left(\int_{t}^{T}R_{\theta}(T-t,T-s)p_{iK}(s)ds\right)dt
\]

the $i-$th columns

Finally the asymptotic variance of $\widehat{\theta}^{T}$ is equal
to: 
\begin{equation}
Var(\widehat{\theta}^{T})=4\frac{\partial^{2}S(X^{*};\theta^{*},\lambda)}{\partial\theta^{T}\partial\theta}^{-1}P(\theta^{*},X^{*})Var(\widehat{\beta}_{K})P(\theta^{*},X^{*})^{T}\frac{\partial^{2}S(X^{*};\theta^{*},\lambda)}{\partial\theta^{T}\partial\theta}^{-1}\label{eq:as_variance_expression}
\end{equation}

and we can use the consistent estimator: 
\[
\widehat{Var(\widehat{\theta}^{T})}=4\frac{\partial^{2}S(\widehat{X};\widehat{\theta}^{T},\lambda)}{\partial\theta^{T}\partial\theta}^{-1}P(\widehat{\theta}^{T},,\widehat{X})\widehat{Var(\widehat{\beta}_{K})}P(\widehat{\theta}^{T},,\widehat{X})^{T}\frac{\partial^{2}S(\widehat{X};\widehat{\theta}^{T},\lambda)}{\partial\theta^{T}\partial\theta}^{-1}
\]

{\footnotesize{}\bibliographystyle{plain}
\bibliography{biblio_weak}
}
\end{document}